\documentclass[11pt,letterpaper]{article}

\usepackage[letterpaper, left=1in, right=1in, top=0.9in, bottom=0.9in]{geometry}
\usepackage[american]{babel}
\usepackage[normalem]{ulem}
\usepackage{amsmath, amssymb, cases, amsthm}
\usepackage{thmtools}
\usepackage[shortlabels]{enumitem}
\usepackage{mdframed}
\usepackage{bbm}
\usepackage{bm}
\usepackage{microtype}
\usepackage{xcolor}
\usepackage{makecell}
\usepackage{mathtools}
\usepackage{algorithmic}
\usepackage[procnumbered,ruled,vlined,linesnumbered]{algorithm2e}
\usepackage{float}
\usepackage{varwidth}
\usepackage{modletter}
\usepackage{tcolorbox}

\usepackage{dsfont}

\usepackage{todonotes}

\SetKwInput{KwData}{Input}
\SetKwInput{KwResult}{Output}
\SetKwInput{KwGlobalVar}{Global variables}

\usepackage{xcolor}
\definecolor{ForestGreen}{rgb}{0.1333,0.5451,0.1333}
\definecolor{DarkRed}{rgb}{0.80,0,0}
\definecolor{Red}{rgb}{1,0,0}
\usepackage[linktocpage=true,
pagebackref=true,colorlinks,
linkcolor=DarkRed,citecolor=ForestGreen,
bookmarks,bookmarksopen,bookmarksnumbered]
{hyperref}

\usepackage[capitalize,noabbrev]{cleveref}

\declaretheorem[numberwithin=section,refname={Theorem,Theorems},Refname={Theorem,Theorems}]{theorem}
\declaretheorem[numberlike=theorem,name=Theorem,refname={Theorem,Theorems},Refname={Theorem,Theorems}]{thm}
\declaretheorem[numberlike=theorem]{lemma}

\declaretheorem{conjecture}
\declaretheorem[numberlike=theorem,name=Corollary,refname={Corollary,Corollaries},Refname={Corollary,Corollaries}]{cor}
\declaretheorem[numberlike=theorem,style=definition]{definition}
\declaretheorem[numberlike=theorem,style=definition,name=Definition,refname={Definition,Definitions},Refname={Definition,Definitions}]{defn}

\declaretheorem[numberlike=theorem,style=remark]{remark}

\declaretheorem[numberlike=theorem, refname={Question,Questions},Refname={Question,Questions},name={Question}]{question}

\declaretheorem[numberlike=theorem, refname={Observation,Observations},Refname={Observation,Observations},name={Observation}]{observation}

\def\final{0}

\ifnum\final=0

\def\gary#1{\marginpar{$\leftarrow$\fbox{GH}}\footnote{$\Rightarrow$~{\sf\textcolor{blue}{#1 --Gary}}}}
\newcommand{\tnote}[1]{{\color{violet}[\textbf{Thatchaphol}: #1]}}
\def\thatchaphol#1{\marginpar{$\leftarrow$\fbox{TS}}\footnote{$\Rightarrow$~{\sf\textcolor{purple}{#1 --Thatchaphol}}}}

\newcommand{\ghnoteinline}[1]{\todo[inline, size=\normalsize, color=blue!20]{Gary's Note: #1}}

\else

\newcommand{\ghnoteinline}[1]{}

\def\gary#1{}
\newcommand{\tnote}[1]{}
\def\thatchaphol#1{}

\fi

\global\long\def\congest{\mathrm{cong}}

\global\long\def\Gtil{\widetilde{G}}
\global\long\def\cG{\mathcal{G}}
\global\long\def\cGtil{\widetilde{\cal{G}}}

\global\long\def\cP{\mathcal{P}}

\newcommand{\eps}{\epsilon}
\newcommand{\poly}{\mathrm{poly}}
\newcommand{\polylog}{\mathrm{polylog}}
\newcommand{\dist}{\mathrm{dist}}

\newcommand{\ignore}[1]{}

\newcommand{\seth}[1]{{\color{red}[{\tiny Seth: \bf #1}]\marginpar{\color{green}*}}}

\title{Reviving Thorup's Shortcut Conjecture}
  \date{}
\author{
Aaron Bernstein\thanks{New York University}
\and
Henry Fleischmann\thanks{Carnegie Mellon University}
\and
Maximilian Probst Gutenberg\thanks{ETH~Zürich}
\and
Bernhard Haeupler\thanks{INSAIT, Sofia University ``St.~Kliment Ohridski'' and ETH~Zürich}
\and
Gary Hoppenworth\thanks{University of Michigan}
\and
Yonggang Jiang\thanks{MPI-INF and Saarland University}
\and
George Z. Li\footnotemark[2]
\and
Seth Pettie\footnotemark[5]
\and
Thatchaphol Saranurak\footnotemark[5]
\and
Leon Schiller\thanks{Hasso Plattner Institute, University of Potsdam}
}

\begin{document}
\sloppy

\begin{titlepage}
  \maketitle 
  \pagenumbering{gobble}
  \begin{abstract}

We aim to revive Thorup's conjecture \cite{thorup1993shortcutting} on 
the existence of \emph{reachability shortcuts} with ideal size-diameter tradeoffs.  Thorup originally asked whether, given any graph $G=(V,E)$ with $m$ edges, we can add $m^{1+o(1)}$ ``shortcut'' edges $E_+$ from the transitive closure $E^*$ of $G$ so that $\dist_{G_+}(u,v) \leq m^{o(1)}$
for all $(u,v)\in E^*$, where $G_+=(V,E\cup E_+)$.
The conjecture was refuted by Hesse \cite{Hesse03}, followed by significant efforts in the last few years to optimize the lower bounds \cite{HuangP21, LuWWX22, bodwin2023folklore, WilliamsXX24, hopp25}.

In this paper, we observe that although Hesse refuted the \emph{letter} of Thorup's conjecture, his work~\cite{Hesse03}---and all followup work~\cite{HuangP21,LuWWX22,bodwin2023folklore,WilliamsXX24,hopp25}---does not refute the \emph{spirit} of the conjecture, which should allow $G_+$ to contain both new shortcut edges and \emph{new Steiner vertices}.
Our results are as follows.

\begin{itemize}
\item On the positive side, we present explicit attacks that break all known shortcut lower bounds~\cite{Hesse03,HuangP21,LuWWX22,bodwin2023folklore, WilliamsXX24,hopp25} when Steiner vertices are allowed.{\small \par}

\item On the negative side, we rule out ideal $m^{1+o(1)}$-size, $m^{o(1)}$-diameter shortcuts whose 
``thickness'' is $t=o(\log n/\log\log n)$, 
meaning no path can contain $t$ consecutive Steiner vertices. 
{\small \par}

\item We propose a candidate hard instance as the next step toward resolving the revised version of Thorup's conjecture. {\small \par}

\end{itemize}

Finally, we show promising implications. Almost-optimal parallel algorithms for computing a generalization of the shortcut that \emph{approximately} preserves distances or flows imply almost-optimal parallel algorithms with $m^{o(1)}$ depth for \emph{exact} shortcut paths and \emph{exact} maximum flow. The state-of-the-art algorithms have much worse depth of $n^{1/2+o(1)}$ \cite{rozhovn2023parallel} and $m^{1+o(1)}$ \cite{chen2022maximum}, respectively.

\end{abstract}

  \setcounter{tocdepth}{3}
  \newpage
  \tableofcontents
  \newpage
\end{titlepage}
\newpage
\pagenumbering{arabic}

\newpage

\section{Introduction }

One of the grand challenges in parallel computing is to overcome the \emph{transitive closure bottleneck}~\cite{KaoK93}: the observation that computing reachability in directed graphs apparently requires \underline{either} superlinear work 
(iterated matrix squaring) to achieve $\polylog(n)$ depth, 
or $n^{\Omega(1)}$ depth to achieve
near-linear work.  
The problem of 
simultaneously 
achieving $m^{1+o(1)}$ work 
and $m^{o(1)}$ 
depth has 
remained an elusive goal for decades,
for reachability and harder problems such as shortest paths and (min-cost) flows.

Naturally, researchers have tried to overcome this bottleneck in various ways.  These efforts have been extremely fruitful, leading to the development of \emph{hopsets} \cite{cohen2000polylog} for computing approximate shortest paths, \emph{shortcuts}~\cite{ullman1990high} to achieve work-depth tradeoffs, min-sketches for estimating transitive closure size~\cite{cohen1997size}, and attention to graph classes where the bottleneck does not exist, e.g., planar graphs~\cite{KaoK93}.
In 1992, Thorup~\cite{thorup1993shortcutting} made a bold, purely combinatorial conjecture that was motivated by the transitive closure bottleneck.

\begin{conjecture}[Thorup's Shortcut Conjecture~\cite{thorup1993shortcutting}]\label{conj:Thorup92}
Let $G^*=(V,E^*)$ be the transitive closure of an  
arbitrary directed graph $G=(V,E)$, with $m=|E|$.
There exists a \emph{shortcut} $E_+\subset E^*$ with $|E_+|\leq m$
such that $G_+=(V,E\cup E_+)$ has $\polylog(n)$ diameter, 
that is, $\dist_{G_+}(u,v) = \polylog(n)$, 
for every $(u,v)\in E^*$.
\end{conjecture}

In other words, once the shortcut $E_+$ is 
discovered in preprocessing, 
single-source reachability could be computed 
near-optimally,
in $\tilde{O}(m)$ work and $\tilde{O}(1)$ depth.
\cref{conj:Thorup92} was known for rooted trees~\cite{Chazelle87} 
with a diameter\footnote{The \emph{diameter} of a directed graph is defined to be the largest \emph{finite} distance.} of $O(\alpha(n,n))$
and later confirmed for planar graphs~\cite{thorup1995shortcutting} with a diameter of $O(\alpha(n,n)\log^2 n)$.\footnote{A corollary of Thorup's planar reachability oracle~\cite{Thorup04} implies a shortcut with size $O(n\log n)$ and diameter 4.}
Thorup~\cite{thorup1993shortcutting} insisted on $|E_+|\leq m$, but one could set any number of thresholds for the shortcut size, e.g., $O(m), m^{1+o(1)}, n^{1.5}, n^{1.99}$, etc.
In 2003 Hesse~\cite{Hesse03} refuted \cref{conj:Thorup92}
and many weaker versions of it.

\begin{theorem}[Hesse~\cite{Hesse03}]\label{thm:Hesse}
    For any $\epsilon>0$ and $n$ sufficiently large, 
    there exist $n$-vertex graphs 
    $G = (V,E)$ and $G'=(V,E')$ with $m,m'$ edges, respectively,
    such that:
    \begin{itemize} 
    \item 
    $m=\Theta(n^{19/17})$ and any 
    shortcut for $G$ with $o(n^{1/17})$ diameter
    must have $\Omega(mn^{1/17})$ edges.
    \item 
    $m'=n^{1+\epsilon}$ and any shortcut for $G'$ with $n^{o(1)}$ diameter must have 
    $\Omega(n^{2-\epsilon})$ edges.
    \end{itemize}
\end{theorem}

\cref{thm:Hesse} precludes the most optimistic outcome, but it still leaves open many fascinating and potentially practical questions, such as: what diameter can be achieved with $\tilde{O}(m)$- or $\tilde{O}(n)$-size shortcuts? And can such shortcuts be computed efficiently?  
It is now known that the optimum diameters for $\tilde{O}(n)$- 
and $\tilde{O}(m)$-size shortcuts are in the ranges 
$[\tilde{\Omega}(n^{1/4}),\tilde{O}(n^{1/3})]$,
and 
$[\tilde{\Omega}(n^{2/9}), \tilde{O}(n^{2/3}/m^{1/3})]$, respectively, and that the upper bounds can be attained in polynomial time. 
See \cref{tab:shortcuts} for details and references.

\newcommand{\istrut}[2][0]{\rule[- #1 mm]{0mm}{#1 mm}\rule{0mm}{#2 mm}}
\begin{table}[]
    \centering
    \begin{tabular}{|l|l|l|l|}
    \multicolumn{1}{l}{\textsf{Shortcut Size}} &
    \multicolumn{1}{l}{\textsf{Diameter}} &
    \multicolumn{1}{l}{\textsf{Construction Time}} &
    \multicolumn{1}{l}{\textsf{Notes/Citation}}\\\hline
      & $O(\sqrt{n})$     & $\tilde{O}(m\sqrt{n})$ & \cite{ullman1990high}\istrut[1]{4}\\\cline{2-4}
     $\tilde{O}(n)$       & $\tilde{O}(n^{2/3})$       & $\tilde{O}(m)$        & \cite{Fineman20} ($\tilde{O}(n^{2/3})$ depth)\istrut[1]{4}\\\cline{2-4}
            & $n^{1/2+o(1)}$            & $\tilde{O}(m)$        & \cite{LiuJS19}  ($n^{1/2+o(1)}$ depth)\istrut[1]{4}\\\cline{2-4}
         &  $O(n^{1/3})$                & $\tilde{O}(mn^{1/3})$ & \cite{kogan2023faster,kogan2022beating,kogan2022new}\istrut[1]{4}\\\cline{2-4}
         &  $\tilde{\Omega}(n^{1/6})$           &                       & \cite{HuangP21}; see also~\cite{Hesse03}\istrut[1]{4}\\\cline{2-4}
         &  $\tilde{\Omega}(n^{1/4})$   && \cite{BodwinH23}\istrut[1]{4}\\\hline\hline
      & $O(n/\sqrt{m})$ & $\tilde{O}(m^{3/2})$ & \cite{ullman1990high}\istrut[1]{4}\\\cline{2-4}
    $\tilde{O}(m)$ & $O(n^{2/3}/m^{1/3})$ & $\tilde{O}(m^{7/6}n^{1/6})$ & \cite{kogan2023faster,kogan2022beating,kogan2022new}\istrut[1]{4}\\\cline{2-4}
    & $\tilde{\Omega}(n^{1/11})$ &  & \cite{HuangP21}; see also~\cite{Hesse03}\istrut[1]{4}\\\cline{2-4}
    & $\tilde{\Omega}(n^{1/8})$ && \cite{LuWWX22}\istrut[1]{4}\\\cline{2-4}
    & $\tilde{\Omega}(n^{1/5})$ &&
\cite{WilliamsXX24}\istrut[1]{4}\\\cline{2-4}
    & $\tilde{\Omega}(n^{2/9})$ &&     
    \cite{hopp25}\istrut[1]{4}\\\hline
    \end{tabular}
    \caption{Upper and lower bounds on diameter specialized to $\tilde{O}(m)$- and $\tilde{O}(n)$-size shortcuts.}
    \label{tab:shortcuts}
\end{table}

\subsection{Reviving \cref{conj:Thorup92}}

The shortcut problem, as Thorup~\cite{thorup1993shortcutting} 
originally defined it, is now well on its way to being completely solved, 
as we have narrowed down the correct exponent to a small range and 
just a couple plausible answers: $1/4$ or $1/3$.

However, \cref{conj:Thorup92} is perhaps \emph{too strict} from the algorithmic perspective.
This is because the final goal is 
to somehow replace a directed graph $G=(V,E)$ with \emph{any} $G_+$ of comparable size such that (1) the transitive closure of $G$ can be extracted from the transitive closure of $G_+$, and (2) every $(u,v)\in E^*$ is witnessed by a 
``short'' $u$-$v$ path in $G_+$.  We propose the following small fix to \cref{conj:Thorup92} by allowing $G_+$ to introduce \emph{Steiner vertices}.

\begin{conjecture}[Steiner Shortcut Conjecture]\label{conj:Steiner}
    Let $G^*=(V,E^*)$ be the transitive closure of a directed graph $G=(V,E)$.
    There exists a graph $G_+ = (V\cup V_+, E\cup E_+)$ such that
    \begin{itemize}
        \item $G_+$ captures the transitive closure of $G$, i.e., for all $(u,v)\in V\times V$, $v$ is reachable from $u$ in $G$ iff it is reachable from $u$ in $G_+$.
        \item $|E_+|=\widetilde{O}(|E|)$;
        \item $G_+$ has $\polylog(n)$ diameter, i.e., $\forall (u,v)\in E^*$, $\dist_{G_+}(u,v) = \polylog(n)$.
    \end{itemize}
\end{conjecture}

We are not the first to propose the use of Steiner points.  Berman, Bhattacharyya, Grigorescu, Raskhodnikova, Woodruff, and Yaroslavtsev~\cite{BermanBGRWY14} proved that
if $G$ is a DAG with dimension $d$, meaning
there is a reachability-preserving embedding
into the oriented hypergrid $[n]^d$, then 
it has a diameter-2 graph $G_+$ with 
$O(n\log^d n)$ shortcuts.  However, the prior work 
on Steiner shortcuts~\cite{BermanBGRWY14,Raskhodnikova10}
has no bearing on \cref{conj:Steiner}.

\subsection{Contributions}

The purpose of this paper is to explore \cref{conj:Steiner}
and the power of Steiner vertices more generally in constructing shortcuts.
We are \emph{uncertain} about the truth of \Cref{conj:Steiner} and our goal is to raise it as an open problem.   
Although 
\cref{conj:Steiner} can be refuted using known techniques when the size of the Steiner shortcut is limited to $|E_+| \le |E|^{1-\epsilon}$ 
for a constant $\epsilon$ (see \Cref{sec:lower bound sublinear m}), once we permit $|E_+| = \widetilde{O}(|E|)$, all known techniques become ineffective.

Indeed, our main contribution is to carefully demonstrate 
why the lower bounds of 
Hesse~\cite{Hesse03}, Huang and Pettie~\cite{HuangP21},
Lu, Williams, Wein, and Xu~\cite{LuWWX22},  
Williams, Xu, and Xu~\cite{WilliamsXX24}, and Hoppenworth, Xu, and Xu~\cite{hopp25} fail 
catastrophically when Steiner vertices are allowed.  
These results are presented in \cref{sec:Hesse-counterexample}.

Our second result, presented in \Cref{sec:thick steiner shortcuts}, can be viewed as either partially 
refuting \cref{conj:Steiner} or  identifying the crucial property that the shortcut for \cref{conj:Steiner} must satisfy.
Specifically, we show that the shortcut must have a sufficiently large \emph{thickness}, defined as follows.

\begin{definition}[Thickness]\label{def:thickness} 
    Let $G_+=(V\cup V_+,E\cup E_+)$ be a Steiner shortcut graph for $G=(V,E)$.
    We say $G_+$ has \emph{thickness} $t$
    (or is \emph{$t$-thick})
    if every path in $G_+$ has at most $t-1$ consecutive Steiner vertices from $V_+$.
\end{definition}

For example, if $V_+=\emptyset$ (there are no Steiner vertices) then the thickness is $1$,
and all prior lower bounds~\cite{Hesse03,HuangP21,LuWWX22,WilliamsXX24,hopp25}
are merely for 1-thick shortcut graphs. 
A 2-thick shortcut graph is one for which the in- and out-neighbors of $V_+$-vertices are all original $V$-vertices.

It turns out that almost all existing lower bounds -- including the state of the art from \cite{WilliamsXX24,hopp25} -- break once we allow Steiner shortcuts of thickness two. Our main result in \cref{sec:thick steiner shortcuts} is to show that thickness two is nonetheless insufficient for general graphs:

\begin{thm} \label{thm:lowerbound-introduction}
For $t\geq 1$, there exists an $m$-edge graph $G$
such that every $t$-thick Steiner shortcut graph $G_+$
for $G$ either has diameter $m^{\Omega(1/t)}$ or $m^{1+\Omega(1/t)}$ edges.  
\end{thm}

In other words, \cref{thm:lowerbound-introduction} refutes \cref{conj:Steiner} for any fixed $t$, and in fact the proof goes through for any $t=o(\log n/\log\log n)$. This implies that any upper bound for \cref{conj:Steiner} would require a somewhat intricate connection between the Steiner vertices.

Roughly speaking, the ``hard'' graphs used to prove \cref{thm:lowerbound-introduction} 
are a family $\mathcal{G}_k$, which are constructed as the $k$-fold product of a base graph,
where $k>t$.  
This graph class is \emph{highly structured}, which actually makes it too easy to shortcut with Steiner vertices; in particular, the $k$-fold product $\mathcal{G}_k$ can be shortcut with thickness $t = k+1$.  A natural and generic way to make
this graph class more difficult to shortcut is to 
discard each edge with some probability $p$ 
inversely proportional to the diameter; denote this noised class as $\cGtil_k(p)$.  It is an open problem whether $\cGtil_k(p)$ suffices to refute \cref{conj:Steiner}. As mild evidence of the hardness of this graph family, we show that, unlike $\mathcal{G}_k$, it is resistant to $t$-thick Steiner shortcuts for 
$t \leq 8k/7$. 
See Section \ref{sec:thick steiner shortcuts} for a more detailed discussion.

\paragraph{Potential Implications.} If Conjecture \ref{conj:Steiner} turns out to be true, this would have exciting algorithmic applications. Most notably, it would suggest a promising approach for solving parallel directed reachability with near-linear work and $\polylog$ depth.

In the directed reachability problem, we are given as input a graph $G$ and a vertex pair $s,t$, and asked to report whether there exists any path from $s$ to $t$. While low-depth algorithms are known if we allow polynomial work, it is still unclear what depth is possible if we restrict ourselves to $m^{1+o(1)}$ work. A depth of $n$ is trivial and was the best known for a long time. A breakthrough of Fineman in 2018 achieved depth $n^{2/3}$ \cite{Fineman20}, which was later improved to $n^{1/2+o(1)}$ \cite{LiuJS19}. The $n^{1/2+o(1)}$ bound can also be generalized to exact shortest paths~\cite{cao2020efficient,RozhonHMGZ23}. 

All known $m^{1+o(1)}$-work algorithms with sublinear depth 
rely on the following basic approach:
{\bf 1)} compute a shortcut $E_+$ with diameter $h$ 
{\bf 2)} run BFS from $s$ in $G_+ = (V,E\cup E')$.  The existing work has focused on optimizing step 1, and there is still potential for further progress in this direction. But such progress has a hard barrier: even with the perfect hopset construction, the lower bounds of \cite{hopp25} proves that this approach could never lead to depth $o(n^{2/9})$. 

Observe, however, that the framework would work just as well if we replaced $E'$ with a \emph{Steiner} shortcut. So if Conjecture \ref{conj:Steiner} turned out to be true, this would suggest a natural approach for overcoming the above barrier and achieving $\polylog(n)$ depth. An analogous framework is also frequently used in other models of computation (dynamic, distributed, streaming), and in all those cases a Steiner shortcut would work just as well as a regular one. 

The main body of the paper focuses exclusively on the simplest case of (Steiner) shortcuts for reachability because even here we are not sure if Conjecture \ref{conj:Steiner} is true or false. But if the conjecture ends up being true, then one can state a natural generalization of the conjecture for shortest paths (a Steiner hopset) and maximum flow (a low-congestion Steiner shortcut): see \Cref{sec:generalized} for more details.

\paragraph{Connection to Circuit Complexity.}
While Steiner shortcuts are primarily motivated by the algorithmic applications mentioned above, they are also a natural combinatorial question in their own right. 
Our \cref{conj:Steiner} can be recast in terms of circuit complexity.  Consider a multi-input, multi-output circuit $C$ consisting solely of unbounded fan-in OR gates.
We would like to ``compress'' $C$ into an equivalent 
circuit of comparable size but smaller depth.
Conjectue \ref{conj:Steiner} is equivalent to \Cref{conj:circuit} below; none of the prior work on depth-reduction~\cite{Hastad86,Schnitger82,Schnitger83,raz1997separation, de2016limited,GolovnevKW21} directly applies to \cref{conj:circuit}.

\begin{conjecture}\label{conj:circuit}
    Suppose $C : \{0,1\}^a \to \{0,1\}^b$ is a circuit
    consisting solely of OR gates and $m$ wires.  
    There is a logically equivalent circuit 
    $C_+ : \{0,1\}^a\to \{0,1\}^b$ that has 
    $\widetilde{O}(m)$ wires and depth $\polylog(m)$.
\end{conjecture}

\paragraph{Organization.}
In \Cref{sec:Hesse-counterexample} we give a self-contained exposition of Hesse's lower bound for (non-Steiner) shortcuts, which we denote by graph family $\cG_k$; we will later use this same graph for \Cref{thm:lowerbound-introduction}. We then show that the lower bound fails once Steiner vertices are allowed. In \Cref{sec:app:breaking} we also show how Steiner vertices break the lower bound of \cite{WilliamsXX24}.\footnote{The most recent $\Omega(n^{2/9})$ lower bound of \cite{hopp25} is a slightly optimized version of \cite{WilliamsXX24} and can be broken in the same way.} 
All other existing lower bounds that we know of can be broken in a similar fashion using Steiner vertices, but we omit the details.

\Cref{sec:thick steiner shortcuts} contains our thickness lower bounds. We prove \cref{thm:lowerbound-introduction}. We also formalize the noised graph family $\cGtil_{k}$---our candidate hard instance for  \Cref{conj:Steiner}---and prove that shortcutting $\cGtil_{k}$ requires strictly higher thickness than the one for $\cG_k$.

In \Cref{sec:generalized} we discuss generalizations of Steiner shortcuts to Steiner hopsets and Steiner low-congestion shortcuts.

\section{Hesse's Counterexample Graphs $\mathcal{G}_k$}\label{sec:Hesse-counterexample}
In this section we review Hesse's lower bound construction~\cite{Hesse03} and illustrate why his argument breaks down when Steiner vertices are allowed.\footnote{Technically, we consider a slight variation of Hesse's construction, decreasing the density of the graph to improve our achievable thickness lower bounds from $\sqrt{\log n}$ to $\log n/\log \log n$. This modification has no significant effect on the analysis in the non-Steiner setting nor for breaking the construction via Steiner points.} The attack presented for Hesse's construction also applies, with suitable modifications, to all subsequent lower bound \cite{HuangP21,LuWWX22,WilliamsXX24,hopp25}; for details of how to do this with \cite{WilliamsXX24} in particular, see Appendix \ref{sec:app:breaking}.
The reason we describe Hesse's counterexample in detail 
is that it is the basis for 
the lower bounds of 
Section \ref{sec:thick steiner shortcuts}.

Following~\cite{Hesse03} we construct a family of graphs $\mathcal{G}_k$, $k \geq 2$.
Each graph $G=(V,E) \in \mathcal{G}_k$ will have an associated set of critical vertex pairs $P \subseteq V \times V$ such that any reachability shortcut of $G$ 
has size at least $|P|$ or hop diameter at least 
$\text{poly}(n)$ between some pair of vertices 
$(s, t) \in P$. The properties of graph family $\mathcal{G}_k$ can be summarized in the following theorem. 

\begin{restatable}[{Cf. \cite[Theorem 4.1]{Hesse03}}]{theorem}{hesse} \label{thm:hesse}
Let $k > 1$ be an integer and $\epsilon=\Omega(1/k)$.  
$\mathcal{G}_k$ is an infinite family of $n$-vertex, $n^{1+O(1/k)}$-edge directed graphs. Each graph $G=(V,E) \in \mathcal{G}_k$ has an associated set of critical vertex pairs $P \subseteq V \times V$ of size 
$|P| = \Omega(n^{1.5-\eps})$.
For any shortcut set $E_+ \subseteq E^*$ of $G$, 
where $G^*=(V,E^*)$ is the transitive closure of $G$,
either
\begin{itemize}
    \item $|E_+| \geq |P|$,  or 
    \item There is a pair of vertices $(s, t) \in P$ 
    such that $\dist_{G_+}(s,t) = n^{\Omega(1/k)}$ in 
    $G_+=(V,E\cup E_+)$.
\end{itemize}
\end{restatable}

In \Cref{subsec:break}, we show that 
Hesse's construction and argument fails in the presence of Steiner vertices, as  captured by the following theorem.

\begin{restatable}{theorem}{Steinerattack} \label{thm:Steinerattack}
Let $k\ge 2$ be an integer, and let $G=(V,E)\in\mathcal{G}_k$ be an $n$-vertex, $n^{1+O(1/k)}$-edge graph 
with associated critical vertex pairs 
$P \subseteq V\times V$.
Then $G$ admits a $k$-thick Steiner shortcut graph 
$G_+=(V\cup V_+,E\cup E_+)$ with $|E_+|=|E|/n^{\Omega(1/k)}$,
such that every $(s,t)\in P$ has $\dist_{G_+}(s,t)=k$.
\end{restatable}

\begin{remark} All existing lower-bound constructions for (non-Steiner) shortcut sets have the same high-level structure: they tailor the construction to produce a set of critical vertex pairs $P$, and then show that these particular pairs are hard to shortcut. The constructions make no claims about other vertex pairs. For this reason, to show that these lower-bound proofs fail in our new setting, it is enough to show that the critical pairs in particular can be shortcut using Steiner vertices. We strongly believe that for all these constructions one can in fact use Steiner vertices to shortcut \emph{all} pairs in the graph, but we did not think this would be instructive, since by design only the critical pairs have interesting structure.
\end{remark}

In \cref{sec:thick steiner shortcuts} we show
that \cref{thm:Steinerattack} is essentially tight for the family $\mathcal{G}_k$, in that any $t$-thick Steiner shortcut set 
with $t\ll k$ must have size superlinear in $m$ or diameter polynomial in $m$.
Before reviewing Hesse's 
construction of $\mathcal{G}_k$
we begin with $\mathcal{G}_1$, 
which introduces some of the key
ideas and is sufficient to 
prove lower bounds on $\tilde{O}(n)$-size shortcuts; 
see Huang and Pettie~\cite{HuangP21}.

\subsection{Warm-up Construction --- Graph Family $\mathcal{G}_1$}\label{subsec:warmup-construction}

We construct a graph family $\mathcal{G}_1$ that does not quite imply polynomial diameter lower bounds for $\tilde{O}(m)$-size shortcut sets. 
This will motivate some of the design choices introduced in the graph family $\mathcal{G}_k$ in \Cref{subsec:counter}. 

Graph family $\mathcal{G}_1$ is parameterized by constants $c \in \bbR^+$ and $d \in \bbN^+$.  We denote a specific parameterization of $\mathcal{G}_1$ by  $ \mathcal{G}_1(c, d)$. 
The size of each graph $G = (V, E) \in \mathcal{G}_1(c, d)$ is controlled by a parameter $r\in \bbN^+$.

    \paragraph{Vertex Set $V$.}  Let $\ell = r^c$.     Graph $G$ has $\ell+1$ layers, $L_0,L_1, \dots, L_{\ell}$. The vertices of each layer can be written as a collection of points in a $d$-dimensional grid. Formally, for $i \in [0, \ell]$, we define $L_i$ to be
    $$
    L_i = \{i\} \times (-r(r+1)\ell, r(r+1)\ell]^d,
    $$
    and let $V = \cup_{i=0}^{\ell} L_i$.
    \paragraph{Edge Set $E$.} Let $B(r, d) \subseteq \bbZ^d$ denote the set of  all integer vectors $\Vec{v}$ in $\mathbb{Z}^d$ of magnitude $\|\Vec{v}\| \leq r$, where $\|\cdot \|$ denotes the standard Euclidean norm. Let $\mathcal{B}(r, d) \subseteq B(r, d)$ denote the set of extreme points of the convex hull of $B(r, d)$. By B\'ar\'any and Larman~\cite{barany1998convex}, 
    $|\mathcal{B}(r, d)|  = \Theta\left( r^{d \cdot \frac{d-1}{d+1}}\right)$.
 The edge set $E$ of $G$  consists of
all $(x,y)\in V\times V$ where
\begin{align*}
    x &= (i,p) \in L_{i},\\
    y &= (i+1, p + \vec{v}) \in L_{i+1},\\
    \text{and } \vec{v} &\in \cB(r,d).
\end{align*}

\paragraph{Critical  Pairs and Critical Paths.} 
We will define a collection of vertex pairs and associated paths in $G$ that are particularly useful in the analysis of $G$. We first specify a set of source nodes $S \subseteq L_0$ in the first layer of $G$ that will be useful for defining our set of critical vertex pairs. Let $S$ be the set
    $$
    S = \{0\} \times (-r^2\ell, r^2\ell]^d.
    $$
    For each vertex $s = (0, p) \in S$ and a vector $\vec{v} \in \mathcal{B}(r, d)$, let $t$ be the vertex $t = (\ell, p + \ell \vec{v})$.  Note that $t$ is guaranteed to lie in $L_\ell$ 
    due to $p\in (-r^2\ell,r^2\ell]^d$ and 
    $\|\vec{v}\|\leq r$.
   We call $(s, t)$ a \textit{critical vertex pair}; let $P \subseteq S \times L_{\ell}$ denote the set of all critical vertex pairs in $G$. 
    Each $(s,t)\in P$ is identified with the \emph{critical path} $P_{s,t}$, defined as follows.
    $$
    P_{s, t} = ((0, p), (1, p+\vec{v}), \dots, (\ell, p + \ell \vec{v})).
    $$
    Let $\mathcal{P} = \{P_{s, t} \mid (s, t) \in P\}$ be the set of critical paths.  

We summarize several properties of graph $G$, critical pairs $P$, and critical paths $\mathcal{P}$ in the following lemma, which we state without proof. 

\begin{lemma}[{Cf.~\cite[Lemma 2.2 and 2.6]{HuangP21}}]\label{lem:HuangP-properties}
    Graph $G$ and paths $\mathcal{P}$ have the following properties:
    \begin{description}
        \item[Sizes.] $|V| = \Theta(r^{c+2d+cd})$, $|E| = \Theta(r^{d \cdot \frac{d-1}{d+1}} \cdot |V|)$, and $|P| = |\mathcal{P}| = \Theta(r^{2d + cd + d \cdot \frac{d-1}{d+1}})$.
        \item[Uniqueness.] Every path $P_{s, t} \in \mathcal{P}$ is a unique $(s, t)$-path in $G$. 
        \item[Disjointness.] Paths in $\mathcal{P}$ are pairwise edge-disjoint. 
    \end{description}
    \label{lem:hp_lowerbound}
\end{lemma}

\paragraph{A remark on the geometry of graph $G$.}  
Vertex set $V$ and paths $\mathcal{P}$ have a very natural geometric interpretation as an arrangement of points and lines in $\bbR^{d+1}$. 
Notice that the vertices of each path $P_{s, t} \in \mathcal{P}$ lie on a line of the form $\{x \cdot (1, \vec{v}) + (0, p) \mid x \in \bbR\}$ in $\bbR^{d+1}$ for $\vec{v} \in \mathcal{B}(r, d)$ and $(0, p) \in S$. Moreover, because $\mathcal{B}(r, d)$ consists only of extreme points, each such direction $\vec{v}$  \emph{cannot} be expressed as a convex combination of other vectors in $\mathcal{B}(r, d)$.
This geometric perspective is useful for interpreting the Uniqueness and Disjointness properties of \Cref{lem:hp_lowerbound}.

\paragraph{Implications and Limitations of $\mathcal{G}_1$.} 
Graph class $\mathcal{G}_1$ already gives an interesting lower bound: any $O(n)$-sized shortcut-set of $\mathcal{G}_1$ has diameter at least $\Omega(\textrm{poly}(n))$. This follows almost immediately from Lemma \ref{lem:HuangP-properties}: Since the critical paths are pairwise disjoint, a shortcut edge can only be relevant to a single critical path, so shortcutting every critical pair would require at least $|P|$ edges; by the construction of $\mathcal{G}_1$, 
we have $|P| = \Omega(n)$.
Unfortunately it is impossible to have $|P|=\Omega(m)$ 
in $\mathcal{G}_1$.  In order to force $|P| \gg m$, Hesse~\cite{Hesse03} defines graph $\mathcal{G}_k$ to be 
a kind of $k$-fold product of $\mathcal{G}_1$.

\paragraph{Optimizing $\mathcal{G}_1$ for Steiner shortcut lower bounds.} Technically, the construction we consider differs slightly from~\cite{Hesse03} and~\cite{HuangP21}, although the difference is essentially aesthetic for the non-Steiner parts of our discussion. We define $L_i = \{i\} \times (-r(r+1)\ell, r(r+1)\ell]^d$ as opposed to $\{i\} \times (-4r\ell, 4r\ell]^d$ as in~\cite{Hesse03}. Consequently, we can define $S$, the starting points of our critical paths, as $\{0\} \times (-r^2 \ell, r^2 \ell]^d$. $S$ is instead defined as $\{0\} \times (-2r\ell, 2r\ell]^d$ in~\cite{Hesse03}. This means that in Hesse's construction only a constant fraction of vertices in the starting layer are the first vertex in a critical path; for our modified construction, all but a vanishing fraction of the vertices are. When we prove~\Cref{thm:lowerbound-introduction} via~\Cref{thm:lowerbound}, this makes the difference between a thickness lower bound of $o(\sqrt{\log n})$ versus $o(\log n/\log \log n)$.

\subsection{Hesse's $\cG_k$ Family}
\label{subsec:counter}

\paragraph{Construction parameters.} 
Graph family $\mathcal{G}_k$ is parameterized by constants $c \in \bbR^+$ and $d \in \bbN^+$.  We denote a specific parameterization of $\mathcal{G}_k$ by  $ \mathcal{G}_k(c, d)$. 
The size of each 
$G = (V, E)$ in $\mathcal{G}_k$ is controlled by 
a parameter $r\in \bbN^+$. 

\paragraph{Vertex Set.}  
Fix $\ell=r^{c}$.
The vertex set $V$ is partitioned into a collection of $k\ell + 1$ disjoint layers 
$\{L_i\}_{i \in [0, k \cdot \ell]}$. 
We will frequently use $L_{i, j}$ to denote 
layer $L_{i\cdot k + j}$, so for example
$L_{i,k}$ is another name for $L_{i+1,0}$.

The vertices of $L_{i,j}$ are identified with tuples $\left((i,j),(p_0,\ldots,p_{k-1})\right)$, where each point $p_0, \dots, p_{k-1} \in \bbZ^d$ is contained in the rectangle $(-(r+1)r\ell, (r+1)r\ell]^d$.
Let $N = |L_{i,j}| = (2(r+1)r\ell)^{kd}$, 
so the total number of vertices is $(k\ell+1)N$.

\paragraph{Critical Subdirections.}
Recall that $B(r,d) = \{\vec{v} \in \mathbb{Z}^d \mid \|\vec{v}\| \le r\}$ is the radius-$r$ ball in dimension $d$ and $\cB(r,d)\subseteq B(r,d)$ is the set of extreme points of the convex hull of $B(r,d)$.
Each vector in $\cB(r,d)$ defines a \emph{critical subdirection} and 
by \cite{barany1998convex} the number
of critical subdirections is
$\Delta = |\cB(r,d)| = \Theta\left(r^{d\cdot\frac{d-1}{d+1}}\right)$.

\paragraph{Edge Set.} The edge set $E$ consists of
all $(x,y)\in V\times V$ where
\begin{align*}
    x &= ((i,j),(p_0,\ldots,p_j,\ldots,p_{k-1})) \in L_{i,j},\\
    y &= ((i,j+1),(p_0,\ldots,p_j + \vec{v},\ldots,p_{k-1})) \in L_{i,j+1},\\
    \text{and } \vec{v} &\in \cB(r,d).
\end{align*}
Naturally $(i,j+1)$ would be interpreted as $(i+1,0)$ when $j+1=k$.  In this way each vertex has out-degree at most $\Delta$.

\paragraph{Critical Directions, Critical Pairs, and Critical Paths.} 
A \textit{critical direction} is a 
$k$-tuple $(\vec{v}_0,\ldots,\vec{v}_{k-1}) \in (\mathcal{B}(r, d))^k$ of critical subdirections,
so there are clearly $\Delta^k$ critical directions. 
Define $S\subset L_{0,0}$ to be a rectangle around the origin.
$$
S = \{ ((0,0),(p_0,\ldots,p_{k-1})) \in L_{0, 0} \;\mid\; \text{$p_i  \in (-r^2\ell, r^2\ell]^d$ for each $i \in [0, k)$} \}.
$$

For each vertex $s = ((0,0),(p_0,\ldots,p_{k-1})) \in S$ and a critical direction $(\vec{v}_0,\ldots,\vec{v}_{k-1}) \in (\mathcal{B}(r, d))^k$, 
let $t$ be the vertex 
$t = ((\ell,0),(p_0+\ell \vec{v}_0,\ldots,p_{k-1}+\ell\vec{v}_{k-1})) \in L_{\ell, 0}$. 
We call $(s, t)$ a \textit{critical vertex pair};
let $P \subseteq S \times L_{\ell,0}$ denote the set of all critical vertex pairs in $G$. 

\begin{restatable}[{Cf.~\cite{Hesse03}}]{lemma}{criticalpairbasic} \label{lem:criticalpairbasic}
    
    Let $G\in \mathcal{G}_k$ with critical pair set $P$.
    \begin{itemize}
        \item $|P|  = \frac{N\Delta^k r^{dk}}{(r+1)^{dk}}$.
        \item Every critical pair $(s,t)\in P$ has a unique $s$-$t$ path $P_{s,t}$ in $G$, called its \emph{critical path}.
        \item If two critical paths  intersect at 
            $u\in L_i$ and $v\in L_{i'}$ then $i'<i+k$. 
    \end{itemize}
\end{restatable}

\Cref{lem:criticalpairbasic} has a similar geometric interpretation as \Cref{lem:HuangP-properties}. The main difference is that the division into subdirections allows for two critical paths to have small overlap, but they still cannot overlap in two layers at distance $k$, which implies the third property. 
See \Cref{app:criticalpairbasic} 
for formal proof.

\begin{definition}[Summary of Parameters]\label{def:summary-parameters}
Let us summarize the 
relevant parameters of graph $G\in \mathcal{G}_k$ with size parameter $r$.

\begin{itemize}
    \item $G$ has $k\ell + 1 = kr^c+1$ layers in total, each with $N=(2(r+1)r\ell)^{kd}$ vertices, hence 
    $|V| = O(k\ell N) = O(k\ell(2(r+1)r\ell)^{kd})$.
    \item The outdegree of each vertex in $G$ is at most 
    $\Delta = \Theta\left(r^{d \cdot \frac{d-1}{d+1}}\right)$, and in fact the average degree of $G$ is $\Theta(\Delta)$, 
    hence $|E| = \Theta(n\Delta) = \Theta(k\ell\Delta N)$. 
    \item The number of critical pairs/critical paths is 
    $|P| = \Theta\left(\frac{N\Delta^k r^{dk}}{(r+1)^{dk}}\right)$ by \Cref{lem:criticalpairbasic}. 
\end{itemize}
\end{definition}

We are now ready to prove that graph $G$ and critical pairs $P$ imply a polynomial lower bound for reachability shortcuts. 
We restate \cref{thm:hesse} for convenience.

\hesse*

\begin{proof}
Let $E_+ \subseteq E^*$ be the set of shortcut edges for $G$, where 
$|E_+| < |P|$. We say that an edge $(u, v) \in E_+$ is \textit{short} if $(u, v) \in L_i \times L_j$ for some choice of $i, j \in [0, k \cdot \ell]$ such that $j < i + k$, and \emph{long} otherwise.

By \Cref{lem:criticalpairbasic}, each long shortcut edge can only be used by one critical
pair in $P$.  Since $|E_+|<|P|$, by the pigeonhole principle there is some $(s,t)\in P$ such that the shortest $s$-$t$ path in $G_+$ uses only original edges from $E$ and \emph{short} edges from $E_+$, hence
$\dist_{G_+}(s,t) \geq k\ell/(k-1) > \ell$.

Treating $k$ as constant, 
the quantitative claims follow by setting 
$\ell = r^c = r^{O(\epsilon)}$ and $d=\Omega(1/\epsilon)$, so that
$\Delta=\Theta(r^{d(d-1)/(d+1)})= r^{(1-\Theta(\epsilon))d}$,
$n=\Theta((2r)^{(2+c)kd}r^c)$, 
$m=O(n\Delta)=r^{(2+c)kd+c+(1-\Omega(\epsilon))d} = n^{1+O(1/k)}$,
$|P|=\Omega\left(\frac{N\Delta^k r^{dk}}{(r+1)^{dk}}\right)$=$r^{(2+c)kd + (1 - \Omega(\varepsilon))kd}$= $\Omega(n^{1.5-\epsilon})$.
\end{proof}

\subsection{Breaking Hesse's Counterexamples using Steiner Vertices}
\label{subsec:break}

In  this section we will show how Steiner vertices break Hesse's counterexample argument. We will achieve this by proving  \Cref{thm:Steinerattack}, 
restated for convenience.

\Steinerattack*

Let $G \in \mathcal{G}_k$ be the given graph,
parameterized by $r \in \bbN^+$.  Each critical pair
$(s,t)\in P$ is of the form $(((0,0),p), ((\ell,0),p+\ell\vec{v}))$, for some 
$p=(p_0,\ldots,p_{k-1})$ 
and critical direction 
$\vec{v}=(\vec{v}_0,\ldots,\vec{v}_{k-1})$. 
Rather than move from $s$ to $t$ by adding 
the $k$ critical subdirections 
in a ``round robin'' format, 
we can add each all at once, 
walking through new Steiner vertices corresponding to points
\[
p \rightarrow (p_0+\ell\vec{v}_0,\ldots,p_{k-1}) \rightarrow
(p_0+\ell\vec{v}_0,p_1+\ell\vec{v}_1,\ldots,p_{k-1}) \rightarrow
(p_0+\ell\vec{v}_0,p_1+\ell\vec{v}_1,p_2+\ell\vec{v}_2,\ldots,p_{k-1}), 
\]
and so on.  We now give the details.

\paragraph{Constructing a Steiner shortcut.} 
Let $G_+ = (V \cup V_+, E \cup E_+)$ be the Steiner shortcut for $G$.  
The vertex set $V_+$ is partitioned into a collection of $k-1$ disjoint layers $S_1, \dots, S_{k-1}$,
where $S_i$ is the set of tuples 
$(i, (p_1, \dots, p_{k-1}))$, each point $p_1, \dots, p_{k-1} \in \bbZ^d$ being contained in the 
rectangle $(-(r+1)r\ell, (r+1)r\ell]^d$. 
Note that by definition, each $S_i$ has the same number
of vertices as each $L_{i,j}$: $|S_i| = |L_{i, j}| = N$. 

Recall that $P\subset S\times L_{\ell,0}$, where $S\subset L_{0,0}$ is in the first layer.
We include in $E_+$ shortcut edges 
$E_0\cup E_1\cup \cdots \cup E_{k-1}$ 
of the following types.
For each $s = ((0, 0), (p_0, \dots, p_{k-1})) \in S$ and 
each subdirection $\vec{v}_0 \in \mathcal{B}(r, d)$, 
we add edges of the following form to $E_0$. 
$$
\Big(((0, 0), (p_0, \dots, p_{k-1})), \; (1, (p_0 + \ell \vec{v}_0, \dots, p_{k-1})\Big) \in S \times S_1
$$
Likewise, for each index $i \in [1, k-2]$, 
we include in $E_i \subseteq S_i \times S_{i+1}$ 
all edges of the following form,
$$
\Big((i, (p_0, \dots, p_{i}, \dots, p_{k-1})),
\; 
(i+1, (p_0, \dots, p_{i} + \ell\vec{v}_i, \dots, p_{k-1}))\Big) \in S_i \times S_{i+1}
$$
where $(i, (p_1, \dots, p_k)) \in S_i$ 
and $\vec{v}_i \in \mathcal{B}(r, d)$ is a critical subdirection.  Finally, $E_{k-1} \subseteq S_{k-1} \times L_{\ell, 0}$ contains all edges of the following form.
$$
\Big((k-1, (p_0, \dots, p_{k-1})),\; 
((\ell, 0), (p_1, \dots,  p_k+\ell\vec{v}_{k-1}))\Big) \in S_{k-1}\times L_{\ell,0},
$$
where $(k-1, (p_1, \dots, p_k)) \in S_{k-1}$ and 
$\vec{v}_{k-1} \in \mathcal{B}(r, d)$
is a critical subdirection.

This completes the construction of our Steiner 
shortcut $G_+$.  We now argue it is a valid Steiner
shortcut (i.e., it does not change the transitive closure
restricted to $V\times V$) and satisfies the size and distance claims of \Cref{thm:Steinerattack}.

\paragraph{Size of $G_+$.} For each $i \in [1, k-1]$,  $|S_i| = N$, so $|V_+| = kN = \Theta(|V| / \ell)$. 
Each vertex in $V_+$ has indegree and outdegree  at most $|\mathcal{B}(r, d)| = \Delta$. 
We conclude that $|E_+| = O(|V_+| \Delta) = O(|V|\Delta / \ell)$. Then $|E_+| = O(|E| / \ell) 
= |E| n^{\Omega(1/(kd))} = |E| / n^{\Omega(1/k)}$ as claimed.

\paragraph{Transitive closure and diameter.} We now verify that the transitive closure of $G_+$ restricted to the original vertices in $G$ is identical to the transitive closure of $G$.  By construction, for any $(s,t)\in V\times V$, any $s$-$t$ path in $G_+$ that uses Steiner vertices $V_+$ must be contained in 
$W = S\times S_1\times S_2\times\cdots\times S_{k-1}\times L_{\ell,0}$. Let $s \in S \subseteq L_{0, 0}$ and 
$t \in L_{\ell, 0}$ be a pair of nodes such that 
there exists an $s$-$t$ path in $W$.  By construction
of $E_+ = E_0\cup\cdots\cup E_{k-1}$, this implies
$s=((0,0),p), t=((\ell,0),p+\ell\vec{v})$, where $\vec{v}=(\vec{v}_0,\ldots,\vec{v}_{k-1})$ are the critical 
subdirections corresponding to the edges along the path.
In other words, $(s,t)\in P$ was already in the transitive closure of $G$.
Conversely, any such $(s,t)\in P$ has a 
length-$k$ path in $W$.
We conclude that $G_+$ is a valid reachability 
shortcut of $G$ and that $\dist_{G_+}(s,t)=k$ for $(s,t)\in P$.

\section{Necessity of Thick Steiner Shortcuts}
\label{sec:thick steiner shortcuts}
This section has two goals. First, we show that if ideal Steiner shortcuts exist with size $m^{1+o(1)}$ and diameter $m^{o(1)}$, 
then they must be \emph{thick} on some graphs in the family 
$\cG_k(c,d)$. 
Second, we prove that a certain \emph{randomized} 
graph class $\widetilde{\cG}_k(c,d,p)$ defined shortly 
is \emph{strictly} harder to shortcut
than $\cG_k(c,d)$, and argue that it may in fact be sufficient to refute \Cref{conj:circuit}.

Our first result is that $\cG_k$ cannot be shortcut 
by any $t$-thick Steiner shortcut unless $t\geq 0.99k$. 

\begin{thm} \label{thm:lowerbound}
For $k\ge 2$, $t\leq 0.99k$, and $m=2^{\omega(k \log k)}$,
every $t$-thick Steiner shortcut
for the $m$-edge graph $G \in \cG_k$ requires either 
$m^{\Omega(1/k)}$ diameter or $m^{1+\Omega(1/k)}$ 
edges.
\end{thm}

\Cref{thm:lowerbound}
rules out $t$-thick Steiner shortcuts
for $t=o( \log n/ \log \log n)$ and will be proved in \Cref{subsec:deterministiclowerbound}.
Note that \Cref{thm:lowerbound}'s constraint on $t$
is essentially sharp, as 
\Cref{thm:Steinerattack} shows that $k$-thick Steiner
shortcuts can shortcut all the critical vertex pairs
with length-$k$ paths.

\paragraph{Noisy Version of $\mathcal{G}_k$.} We now introduce a noisy variant of the graph family $\mathcal{G}_k$, which we call $\cGtil_k(c,d,p)$. It is defined as $\cGtil_k(c,d,p) = \{ \widetilde{G}(p) \mid G\in  \cG_k(c,d)\}$, where $\widetilde{G}(p)$ denotes the (random) graph obtained from a given $G$ when deleting each edge independently with probability $p$. 
$\cGtil_k(c,d,p)$ is thus a family of random graphs that we obtain from the elements of $\cG_k(c,d)$ by deleting each edge independently with probability $p$.
Here $p$ should not be thought of as an absolute constant, but a function of $c,d,r$ that is inversely proportional to the diameter of $G$.

Intuitively, the problem with $\cG_k$ is that it has an extremely rigid structure that can be exploited by Steiner shortcuts. By randomly deleting edges, $\widetilde{\cG}_k(c,d,p)$ destroys this structure.
We highlight the following open problem,
the resolution of which would either refute \Cref{conj:Steiner} or contribute a significant step toward computing almost optimal Steiner shortcuts, 
either of which would be very exciting. 

\begin{question}\label{qu:Gcdp-shortcut}
Are there parameters $k\geq 2,c,d,p$ 
for which any Steiner shortcut for an $m$-edge 
graph $G \in  \cGtil_{k}(c,d,p)$ cannot simultaneously
have size $m^{1+o(1)}$ and diameter $m^{o(1)}$?
\end{question}

We report some partial progress towards 
answering \Cref{qu:Gcdp-shortcut} by showing, 
in a formal sense, that $\cGtil_{k}(c,d,p)$ is indeed harder to shortcut than $\cG_{k}(c,d)$.
The following theorem, which will be proved in  \Cref{subsec:randomizedlowerbound}, shows that $\widetilde{\cG}_k(c,d,p)$ 
requires $t$-thick Steiner shortcuts for $t>(1+\Omega(1))k$. In contrast, \Cref{thm:lowerbound} applies only when $t \le 0.99k$. 

\begin{thm}\label{thm:improvedlowerboundundernoise}
For $k\geq 2$, there are parameters $c,d,p$
such that with probability $1-1/\poly(m)$, 
every $t$-thick Steiner shortcut for the $m$-edge graph 
$G \in \cGtil_{k}(c,d,p)$ requires either 
diameter $m^{\Omega(1/k)}$ or $m^{1+\Omega(1/k)}$ 
edges, provided that
$1.01t \le \frac{8}{7}k$ and 
$m=2^{\omega(k \log k)}$.
\end{thm}

\subsection{Basic Properties}

We begin with some basic properties of $\cG_{k}$ 
that are helpful for proving lower bounds. 
Let $G\in\cG_{k}$ be a graph in this graph family with 
the parameters $r,\ell$, etc.~defined as usual; see \Cref{def:summary-parameters}.

Suppose $G_+=(V\cup V_+,E\cup E_+)$ 
is a Steiner shortcut graph for $G$.
For each critical pair $(s,t)\in P$, designate one
shortest $s$-$t$ path in $G_+$ a \emph{critical shortest path}.\footnote{Note that $G$ contains a unique critical path from $s$ to $t$ whereas $G_+$ may contain many $s$-$t$ paths.}
We call the shortcut edge set $E_+$ \emph{normalized}
if it is contained in 
$(L_i\times V_+) \cup (V_+\times V_+) \cup (V_+ \times L_j)$, for two layers $L_i,L_j$, $i+k\leq j$.
\Cref{lem:normalizing} shows that any shortcut edge set $E_+$ can be replaced by the union of $O((k\ell)^2)$ normalized shortcut sets, without increasing the diameter too much.

\begin{lemma}\label{lem:normalizing}
    Let $G\in \cG_{k}$ and $G_+$ be a $t$-thick Steiner shortcut for $G$ with $M$ edges and diameter $h$.
    There is a $t$-thick Steiner shortcut 
    $G'=(V\cup V',E\cup E')$ for $G$ such that
    $E'$ is the union of less than $(\ell k)^2$ normalized $t$-thick shortcut sets, each with at most $M$ edges.
    The diameter of $G'$ is less than $hk$.
\end{lemma}
\begin{proof}
    For $i,j$ such that $i+k\leq j$, let $E_{i,j}'$ 
    be a copy of $E_+$ restricted to 
    $(L_i\times V_+) \cup (V_+\times V_+) \cup (V_+\times L_j)$.
    Let $E'$ be the union of all such $E_{i,j}'$, where
    each uses a disjoint copy of the Steiner vertices $V_+$.
    Every path in $G_+$ from $L_i$ to $L_j$ via $V_+$ is preserved
    by an equivalent path in $G'$, so long as $i+k\leq j$.  Thus, given a
    length-$h$ path in $G_+$, we can find a corresponding 
    length-$h(k-1)$ path in $G'$, as every subpath 
    from $L_i$ to $L_j$, via zero or more hops in $V_+$, 
    with $j<i+k$
    may need to be replaced by a length-$(j-i)$ path.
\end{proof}

\paragraph{Efficiency.}
Let $G\in\cG_k$ and $G'=(V\cup V',E\cup E')$ be a 
Steiner shortcut for $G$ such that $E' = \bigcup E_{i,j}'$
is the union of disjoint normalized shortcut sets $E_{i,j}'$.
We say that a critical shortest path $\Pi$ in $G'$ is 
\emph{$E_{i,j}'$-respecting} if a subpath of $\Pi$
from $L_i$ to $L_j$ is contained in $E_{i,j}'$.

\begin{definition}[Efficiency]\label{def:efficiency}
Fix $G,G'=(V\cup V',E\cup E'),$ and $E_{i,j}'$ as above, 
and fix the set of critical shortest paths.
\begin{itemize}
    \item The \emph{efficiency of $E_{i,j}'$} is the ratio
    of the number of $E_{i,j}'$-respecting critical shortest paths to $|E_{i,j}'|$.
    \item The \emph{efficiency of a Steiner vertex $u$} incident
    to $E_{i,j}'$ is the ratio between the number of $E_{i,j}'$-respecting critical shortest paths passing through $u$
    to $\deg_{E_{i,j}'}(u)$.
\end{itemize}
\end{definition}

\Cref{lem:efficiencynotdecrease} is a simple observation 
that follows from the definition of efficiency.

\begin{lemma}\label{lem:efficiencynotdecrease}
    Let $G\in\cG_k$ and $G'=(V\cup V',E\cup E')$ be a Steiner shortcut, and $E'_{i,j}\subset E'$ be a normalized Steiner shortcut set from $L_i$ to $L_j$.  
    Suppose $u$ is a Steiner vertex incident to $E_{i,j}'$,
    whose efficiency is at most the efficiency of $E_{i,j}'$.
    Obtain $E_{i,j}''$ by deleting $u$ and its incident edges,
    and rerouting critical shortest paths that used $u$.
    Then the efficiency of $E_{i,j}''$ is at least that of $E'_{i,j}$.
\end{lemma}
\begin{proof}
    Suppose there are $\alpha$ many $E'_{i,j}$-respecting critical shortest paths passing through $u$ and 
    $\deg_{E_{i,j}'}(u)=\beta$. 
    Suppose there are $\alpha'$ total $E'_{i,j}$-respecting critical shortest paths and $|E'_{i,j}|=\beta'$. 
    By \Cref{def:efficiency}, $E'_{i,j}$ has efficiency 
    $\alpha'/\beta'$, and, after deleting $u$, $E_{i,j}''$ has an efficiency of at least $(\alpha'-\alpha)/(\beta'-\beta)\ge \alpha'/\beta'$,
    since $\alpha/\beta \le \alpha'/\beta'$.
\end{proof}

The following lemma is crucial for our analysis.

\begin{lemma}\label{lem:volumnbound}
    Recall from \Cref{def:summary-parameters} that $\Delta = \Theta(r^{d\frac{d-1}{d+1}})$ is the number of critical subdirections. Given a graph $G\in\cG_{k}$ and a 
    Steiner shortcut $G'=(V\cup V', E\cup E')$ for $G$, the number of critical shortest paths passing through 
    any Steiner vertex $u\in V'$ is at most $\Delta^k$.
\end{lemma}
\begin{proof}
    By the pigeonhole principle, if there are more than $\Delta^k$ critical shortest paths passing through $u$
    then two are associated with the same critical direction,
    say $(\vec{v}_0,\ldots,\vec{v}_{k-1})$. 
    Let the two critical pairs be $(a,a'),(b,b')\in P$,
    where
    \begin{align*}
        a &= ((0,0),(a_0,\ldots,a_{k-1}))\\
        b &= ((0,0),(b_0,\ldots,b_{k-1}))\\
        a'&= ((\ell,0),(a_0+\ell \vec{v}_0,\ldots,a_{k-1}+\ell \vec{v}_{k-1}))\\
        b'&= ((\ell,0),(b_0+\ell \vec{v}_{0},\ldots,b_{k-1}+\ell\vec{v}_{k-1}))
    \end{align*}
    Since $(a,a')\neq (b,b')$, 
    there must exist an index $i^\star$ such that $a_{i^\star}\not=b_{i^\star}$.
    Since the two critical shortest paths go through $u$ 
    and the Steiner shortcut does not change the transitive closure on $V$, it follows that $(a,b'),(b,a')$ are also in the transitive closure $E^*$ of $G$.
    For $a$ to reach $b'$ 
    there must exist $\ell$ critical subdirections $\vec{u}_0,\vec{u}_1,\ldots,\vec{u}_{\ell-1}$ such that 
    \[
    a_{i^\star}+\sum_{j\in[0,\ell)} \vec{u}_j=b'_{i^\star}.
    \]
    Similarly, for $b$ to reach $a'$, there must exist $\ell$ critical subdirections $\vec{u}_0',\vec{u}_1',\ldots,\vec{u}_{\ell-1}'$ such that 
    \[
    b_{i^\star}+\sum_{j\in[0,\ell)}\vec{u}_j'=a'_{i^\star}.
    \]
    Recall that $a'_{i^\star}-a_{i^\star}=b'_{i^\star}-b_{i^\star}=\ell\vec{v}_{i^\star}$. Combining the equalities above, we get that
    \[
    \sum_{j\in [0,\ell)}(\vec{u}_j+\vec{u}_j') = 2\ell\vec{v}_{i^\star}.
    \]
    However, according to the definition of $\vec{u}_j,\vec{u}_j',\vec{v}_{i^\star}\in \cB(r,d)$ (they are the extreme points of a convex hull), 
    we must have $\vec{u}_j=\vec{u}_j'=\vec{v}_{i^{\star}}$ for every $j\in[0,\ell)$, for otherwise $\vec{v}_{i^\star}=\frac{1}{2\ell}\sum_{j\in [0,\ell)}(\vec{u}_j+\vec{u}_j')$, contradicting the fact that 
    $\vec{v}_{i^\star}$ is a extreme point, and
    contradicting the fact that $a_{i^\star}\not=b_{i^\star}$.
\end{proof}

This section is intended to prove that a normalized Steiner shortcut cannot have a large efficiency.

\begin{lemma}\label{lem:efficientofnormalized}
    Fix a $k \ge 2,\epsilon > 0$. Let $G\in \cG_k$ and 
    $G'=(V\cup V',E\cup E')$ be a $t$-thick Steiner shortcut for $G$, where $t\leq (1-\epsilon)k$.
    Any normalized shortcut set 
    $E'_{i,j} \subseteq E'$ has efficiency 
    at most $\Delta^{k-1-\eps}$.
\end{lemma}
\begin{proof}
    Suppose that $E'_{i,j}$ has efficiency greater than
    $\Delta^{k-1-\eps}$.
    Begin by deleting all Steiner vertices from 
    $E'_{i,j}$ whose degrees are at least $\Delta^{1+\eps}$,
    and let $E_{i,j}''$ be the resulting edge set.
    Notice that according to~\cref{lem:volumnbound}, 
    these vertices have efficiency at most $\Delta^k/\Delta^{1+\eps}=\Delta^{k-1-\eps}$. 
    According to \Cref{lem:efficiencynotdecrease}, the efficiency of $E_{i,j}''$ is at least as large as that of $E'_{i,j}$.
    
    Recall that the efficiency of $E_{i,j}''$ is the ratio of the number of $E_{i,j}''$-respecting critical shortest paths to $|E_{i,j}''|$.  
    Thus, there must exist an edge $e\in E_{i,j}''$ 
    such that at least
    $\Delta^{k-1-\eps}$ different $E_{i,j}''$-respecting critical shortest paths use $e$, and moreover, use $e$ as the \emph{first} edge in $E_{i,j}''$.
    Due to the $t$-thickness of $E_{i,j}''$, 
    every path in $E_{i,j}''$ has length at most 
    $t \leq (1-\eps)k$, which, with the degree bound 
    $\Delta^{1+\epsilon}$, implies that the total number
    of paths in $E_{i,j}''$ using $e$ as the first edge
    is at most
    \[
    \Delta^{(t-1)(1+\epsilon)} 
    \leq 
    \Delta^{((1-\epsilon)k-1)(1+\epsilon)}
    <
    \Delta^{k-1-\epsilon}. 
    \]
    By the pigeonhole principle, there must be two 
    $E_{i,j}''$-respecting critical shortest paths
    $\Pi_1,\Pi_2$
    that take the same route through $E_{i,j}''$, 
    beginning with edge $e$.  This implies that 
    $\Pi_1,\Pi_2$ share the same vertices $v_i\in L_i$ and $v_j\in L_j$.  Notice that $v_i,v_j$ must also be in the two corresponding critical paths in $G$,
    which follows from
    the uniqueness of critical paths (\Cref{lem:criticalpairbasic}).
    By the definition of the normalized Steiner shortcut $E_{i,j}'$, $j\geq i+k$, and by the limited intersection property of \Cref{lem:criticalpairbasic}, 
    this implies that $\Pi_1,\Pi_2$
    are not distinct paths, but 
    associated with the \emph{same}
    critical pair in $P$, a contradiction.
\end{proof}

\subsection{$\cG_{k}$ Needs Thick Steiner 
Shortcuts---Proof of~\cref{thm:lowerbound}}\label{subsec:deterministiclowerbound}

\begin{proof}[Proof of~\cref{thm:lowerbound}]
    Select $G\in \cG_{k}=\cG_k(c,d)$, 
    where $c=3,d=20000$.  Let $m$ and $n$ be the number of edges and vertices in $G$, where $m=\exp(\omega(k \log k))$. 
    Suppose there exists a Steiner shortcut $G_+=(V\cup V_+,E\cup E_+)$ for $G$ with 
    $|E_+| \leq m^{1+\eps^2_0/k}$ edges 
    and diameter $m^{\eps_0/k}$,
    for a sufficiently small constant 
    $\eps_0=\eps_0(c,d)$.  
    We will deduce a contradiction. 
    Recall from \Cref{def:summary-parameters} that we have 
    \[
    m=O(Nk\ell \Delta) \leq r^{k\cdot f(c,d)},
    \]
    where $f(c,d)$ is a constant depending on $c,d$.
    We set $\epsilon_0 < 1/f(c,d)$, so
    \[
    |E_+| = m^{1+\eps^2_0/k} \le m\cdot r^{\eps_0}.
    \]
    We now substitute for $G_+$ a normalized shortcut graph
    $G'=(V\cup V',E\cup E')$, 
    where $E' \bigcup_{i,j} E_{i,j}'$ is the union of less than $(k\ell)^2$ normalized shortcuts $E_{i,j}'$
    from $L_i$ to $L_j$.
    By~\cref{lem:normalizing}, $|E'| < (k\ell)^2|E_+|$
    and $G'$ has diameter upper bounded
    \[
    k\cdot m^{\eps_0 /k}=O(k\cdot r^{\eps_0\cdot f(c,d)})\le O(k\cdot r^{c/2}) < k\ell.
    \]
    Since the diameter is less than the number of layers of $G$ (which is $k\ell$), every critical pair must be connected by a critical shortest path $\Pi$ that uses some $E'$ edge, and therefore must be
    $E_{i,j}'$-respecting for at least one $E_{i,j}'$.  
    Recall from \Cref{def:summary-parameters} that the number of critical pairs is
    \[
    \Omega\left(\frac{N\Delta^kr^{dk}}{(r+1)^{dk}}\right)=\Omega(m/k\ell)\cdot \Delta^{k-1}=\Omega(m)\cdot \Delta^{k-1-\frac{2c}{d}}.
    \] 
    The first inequality follows from $(1 - 1/(r+1))^{dk} = \Omega(1)$ since $r = \Omega(k)$. This in turn is because $r \geq m^{\frac{1}{k \cdot g(c,d)}}$, where $g(c,d)$ is a constant depending on $c$ and $d$, so we have $g(c,d) \cdot k \log r \geq  \log m$. Then since $m=\exp(\omega(k \log k))$, we have $\log m = \omega( \log k)$.  
    
    Next, remember that $E'$ is partitioned into 
    less than $(\ell k)^2\leq (r^{2c}\cdot k^2)$ 
    normalized components $\{E_{i,j}'\}$.
    Thus, there exists a component $E'_{i,j}$ for which $\Omega(m)\cdot\Delta^{k-1-\frac{6{c}}{d}}/k^2$ critical shortest paths are $E'_{i,j}$-respecting. The efficiency of $E'_{i,j}$ is therefore at least
    \[
    \frac{\Omega(m)\cdot \Delta^{k-1-\frac{6{c}}{d}}}{k^2\cdot m\cdot r^{\eps_0}}
    \geq 
    \Omega(\Delta^{k-1-\frac{8{c}}{d}}/ k^2)
    \ge 
    \Omega(\Delta^{k-1-\frac{9{c}}{d}}) 
    > 
    \Delta^{k-1.01}.
    \]
    
    The second to last inequality is because $k= O(r)$. 
    The last inequality follows by plugging in $c=3,d=20000$ 
    taking $\Delta$ to be sufficiently large.
    We now apply \Cref{lem:efficientofnormalized} with $\eps=0.01$, which
    implies that the efficiency of
    $E_{i,j}'$ is at most 
    $\Delta^{k-1.01}$, a contradiction.
\end{proof}

\subsection{$\cGtil_{k}$ Needs Thicker Steiner Shortcuts---Proof of \Cref{thm:improvedlowerboundundernoise}}\label{subsec:randomizedlowerbound}

In this subsection, we improve the bounds from \Cref{subsec:deterministiclowerbound} under the assumption that edges in $G_k$ are deleted independently with probability $p = 1/(k\ell)$. That is, we consider the (random) graphs from the family $\cGtil(c, d, p)$, where in this section $c=3,d=20000$. The introduction of this type of ``noise'' implies that each critical path remains present with constant probability but at the same time it limits the maximum efficiency of Steiner shortcuts which enables improved bounds on their thickness. Precisely, we show that after noising, no sufficiently efficient Steiner shortcut can survive. To this end, given any $A \subseteq L_i, B\subseteq L_j$, we relate the number of critical paths going from $A$ to $B$ to the number of edges on these paths (i.e. their \emph{support size}). Concretely, defining $\sigma := \max\{|A|, |B|\}$, we show that whenever there are $\gg \sigma^{3/2}\polylog(n)$ critical paths going from $A$ to $B$, their support size is so large that with high probability at least one edge supporting our paths is deleted when applying noise. Using a union bound, this in turn implies that the number of critical paths between any pair of sets $A \subseteq V_i, B\subseteq V_{j}$ is at most $\sigma^{3/2}$, which is much smaller than the trivial bound of $\sigma^2$. This limits the maximum efficiency any Steiner shortcut of a given thickness for $G \in \cGtil_{k}(c,d,p)$ with high probability, and allows us to rule out shortcuts of thickness $(8/7)k$ instead of just $k$ as in \Cref{thm:lowerbound}. 

Formally, our proof relies on the following two lemmas, the first of which connects the number of critical paths between $A$ and $B$ to their support size in a $G \in \cG_k(c,d)$.

\begin{lemma}\label{lem:supportsize}
    Consider any $G \in \cG_k(c,d)$ and let $i, j$ be such that $|i - j| \ge 4rk$. Futher, let $A \subseteq L_i, B \subseteq L_j$ be two subsets of vertices in layers $i$ and $j$, respectively. Then, if there is a set $\cP$ of at least $\frac{s^{3/2}}{2\sqrt{rk}}$ critical paths that connect a vertex in $A$ to a vertex in $B$, the number of edges supported by the paths in $\cP$ between the layers  $L_i$ and $L_j$ is at least $s$.
\end{lemma}
\begin{proof}
    Assume that the number of critical paths from $A$ to $B$ is at least $\eta s$ for some suitable $\eta$ ($\eta$ can be seen as the efficiency of a star shortcutting $A$ to $B$). Suppose that the support size of these paths at most $s$ across the layers $L_i, L_{i+1}, \ldots, L_{j}$. We show that a contradiction arises if $\eta \ge \frac{\sqrt{s}}{2\sqrt{rk}}$. This shows that many critical paths between $A,B$ imply a large support size.
    
    Denote the set of critical paths with one vertex in $A$ and one in $B$ by $\cP$. Note that our assumption on the support size of the paths in $\cP$ implies that the average number of outgoing edges associated to any path in $\cP$ per layer in between $A$ and $B$ in $\Gtil$ is at most $s/|i-j|$, so there must be a layer $L_t, t \in [i, j-1]$ with at most $s/(4rk)$ supported outgoing edges incident to $L_t$. We let $U_1$ be the set of vertices in $L_t$ that are on a path in $\cP$ and note that $|U_1| \le s/(4rk)$ because the number of supported vertices is at most the number of supported outgoing edges of $L_t$. 
    
    Since by assumption, there are $\eta s$ paths and each path passes through $U_1$ while $|U_1| \le s/(4rk)$, the average number of paths passing through a vertex in $U_1$ is at least $4 \eta rk$. Now, consider layer $L_{t-k}$ or $L_{t+k}$
    , i.e., a layer prior to or subsequent to layer $L_t$ by $k$ steps. Denote the set of supported vertices in this layer by $U_2$. We consider the logical bipartite graph $H = ( U_1 \sqcup U_2, E_H  )$ obtained by connecting every vertex $u \in U_1$ to all vertices in $U_2$ that lie on a path in $\cP$ passing through $u$. We denote the set of neighbors of $u$ in $H$ by $N_H(u)$ and its degree by $\deg_H(u) \coloneqq |N_H(u)|$. For two vertices $u_1, u_2$, we write $u_1 \sim_H u_2$ to denote that $u_1, u_2$ are adjacent in $H$.
    
    For each vertex $u_1 \in U_1$, we define $p(u_1) \coloneqq \sum_{u_2 \in N_H(u_1)} \deg_H(u_2)$ such that $p(u_1)$ is equal to the total number of critical paths passing through the neighbors of $u_1$ in $H$. We show that there is a vertex $u_1 \in U_1$ with $p(u_1) \ge 4rk\eta^2$.
    To this end, note that \begin{align*}
        \sum_{u_1 \in U_1} p(u_1) &= \sum_{u_1 \in U_1} \sum_{u_2 \in N_H(u_1)} \deg_H(u_2) \\
        &= \sum_{u_1 \in U_1} \sum_{u_2 \in U_2} \mathds{1}(u_1 \sim_H u_2) \deg_H(u_2)\\ 
        &= \sum_{u_1 \in U_1} \sum_{u_2 \in U_2} \sum_{u_1' \in U_1} \mathds{1}(u_1 \sim_H u_2) \mathds{1}(u_2 \sim_H u_1')\\
        &=\sum_{u_2 \in U_2} \deg_H(u_2)^2.
    \end{align*}
    Now, applying Cauchy-Schwarz, \begin{align*}
        \left( \sum_{u_2 \in U_2} \deg_H(u_2)^2 \right) \cdot |U_2| \ge \left( \sum_{u_2 \in U_2} \deg_H(u_2) \right)^2,
    \end{align*} and since $\sum_{u_2 \in U_2} \deg_H(u_2) = \eta s$ and $|U_2| \le s $ because we assumed that the total support of $\cP$ is at most $s$, we get \begin{align*}
         \sum_{u_1 \in U_1} p(u_1) = \sum_{u_2 \in U_2} \deg_H(u)^2 \ge \frac{1}{s} \left( \sum_{u \in U'} \deg_H(u) \right)^2 = \eta^2 s.
    \end{align*} This in turn implies that the average of $p(u_1)$ over $u_1 \in U_1$ is \begin{align*}
        \frac{1}{|U_1|} \sum_{u_1 \in U_1} p(u_1) \ge 4rk \eta^2,
    \end{align*} because $|U_1| \le s/(4 rk)$. We conclude that there exists some $u \in U_1$ such that $p(u) \ge 4rk \eta^2$.

    Note that the vertices corresponding to $N_H(u)$ have a pairwise Euclidean distance within the $dk$--dimensional grid associated to layer $U$ of at most $2rk$ since all of them are reachable from $u$ by following a critical direction, each of which has total Euclidean norm $\le rk$. We claim that following all critical paths incident to the vertices in $N_H(u)$ for $> 2rk$ steps, we reach $p(u)$ distinct vertices. Indeed, if there were a vertex $v$ reachable from two distinct critical paths starting at two vertices $u_1, u_2 \in N_H(u)$ after $> 2rk$ steps, then $\|\vec{u_1} - \vec{u_2}\|_2 > 2rk$  contradicting the fact that the Euclidean distance between $u_1, u_2$ is at most $2rk$ as required by the assumption that $u_1, u_2$ are both neighbors of $u$ in $H$.\footnote{For any $u \in L_i$, let $\vec{u}$ denote the position of $u$ in the $dk$-dimensional integer grid associated with layer $L_i$.} This uses that the critical directions are distinct vectors with integer coordinates. Accordingly, if $|i - j| \ge 4rk$, then \begin{align*}
        4rk\eta^2 \le s
    \end{align*} as otherwise, we obtain a contradiction to the assumption that all paths are supported by at most $s$ edges. 
    Rearranging the above yields that that whenever $\eta > \frac{\sqrt{s}}{2\sqrt{rk}}$, there is some layer on which all critical paths connecting $A$ to $B$ are supported by at least $s$ edges. 
\end{proof}

With the above, $(k\ell\sigma \log(n))^{\frac{3}{2 }}$ critical paths between $A$ and $B$ yield a support size of at least $k \ell \sigma \log(n)$. Using that each edge is deleted independently with probability $p = 1/(k\ell)$ then implies that, for every such set, at least one critical path fails.

\begin{lemma}\label{lem:efficiencybound}
    Consider any $G \in G_k$ and let $i, j$ be such that $|i - j| \ge 4rk$. Let $A \subseteq L_i, B \subseteq L_j$ and define $\sigma \coloneqq \max\{|A|, |B|\}$. Then with probability $1 - 1/\poly(n)$, for all such $i, j, A, B$ in $\Gtil$ with at least $(100 \sigma \log(n) \ell k)^{3/2}$ critical paths connecting a vertex in $A$ to a vertex in $B$, there is at least one such critical path that fails when passing from $G$ to the noised version $\widetilde{G}$.
\end{lemma}
\begin{proof}    
    Let $\gamma := 100 \log n$. We show that for every choice of $A \subseteq L_i, B \subseteq L_j$ with more than  $(\gamma \sigma k \ell)^{3/2}$ critical paths, at least one path will fail after applying noise. Invoking~\Cref{lem:supportsize}, the support size of these paths between layers $L_i$ and $L_j$ is at least $\gamma \sigma k \ell$. Furthermore, the number of all $A, B$ with $\max\{|A|, |B|\} = \sigma$ is at most $n^{2\sigma}$, and the probability that for fixed sets $A, B$ no critical path connecting $A$ to $B$ fails is at most \begin{align*}
        (1-p)^{\gamma \sigma k \ell} = \left( 1 - \frac{1}{k \ell}\right)^{\gamma \sigma k \ell} \le \exp(-\gamma \sigma).
    \end{align*} because we have chosen $p = \frac{1}{k\ell}$. Applying a union bound yields that the probability that there is some $A \subseteq L_i, B \subseteq L_j$ with $|i-j| \ge 4rk$ and at least $(\gamma \sigma k \ell)^{3/2}$ critical paths connecting $A$ and $B$ is at most \begin{align*}
        \sum_{\sigma = 1}^{\infty} (\ell k)^2 n^{2\sigma} \left(1 - \frac{1}{\ell k}\right)^{\gamma \sigma k \ell} &\le (\ell k)^2 \sum_{\sigma = 1}^{\infty} \exp\left(2\sigma \log(n) - \gamma \sigma \right)\\
        &\le (\ell k)^2 \sum_{\sigma = 1}^{\infty} \exp\left(2 \sigma \log(n) - 100 \sigma \log(n)\right)\\
        &\le (\ell k)^2 \sum_{\sigma = 1}^{\infty} n^{-98 \sigma} \le n^2 n^{-97} \le n^{-95}.
    \end{align*} This concludes that with probability $1 - 1/\poly(n)$, there is no pair of sets $A \subseteq L_i, B \subseteq L_j$ such that $|i-j| \ge 4rk$ and such that there are more than $(100 s \log(n) \ell k)^{3/2}$ critical paths connecting $A$ and $B$.

\end{proof}

\begin{proof}[Proof of \Cref{thm:improvedlowerboundundernoise}]
    We show that there is no Steiner shortcut with diameter $m^{\epsilon^2/k}$ while using at most $m^{1 + \epsilon^2/k}$ edges and using thickness 
    $t$ such that $1.01 t \le \frac{8}{7}k$. Here, $\epsilon > 0$ is a sufficiently small constant fixed later.

    To this end, we suppose that there is a Steiner shortcut $E'$ with the desired properties and then deduce a contradiction. We consider the family of random graphs $\cGtil_{k}$ and set $p = 1/(\ell k)$, and $c = 3$ such that $\ell = r^3$. 
    Now, using that $k = O(r)$, the number of critical pairs is \begin{align*}
       \Theta\left(\frac{N\Delta^k r^{dk}}{(r+1)^{dk}}\right) = \Theta(N\Delta) \cdot \Delta^{k-1} = \Theta(m/(k\ell)) \cdot \Delta^{k-1}.
    \end{align*} Recalling that $\Delta = \Theta(r^{d\frac{d-1}{d+1}})$ allows us to set $d = 2/\epsilon$ such that $r = O(\Delta^{2\epsilon})$ implying that the number of critical pairs becomes $\Omega(m) \Delta^{k-1-3\epsilon}$. Furthermore, the number of edges in our graph is at most \begin{align*}
        m = O(Nk\ell \Delta) = O((2(r+1)r\ell)^{kd} k\ell r^{2/\epsilon} )  = O(r^{7k/\epsilon}).
    \end{align*}
    
    Now, we convert the shortcut $E'$ into the union $E''$ of $O((\ell k)^2)$ normalized shortcuts according to \Cref{lem:normalizing}, each with at most $m^{1 + \epsilon^2/k}$ edges. According to the definition of thickness and by splitting edges by Steiner vertices, we can further assume that all components of $E''$ consist of $t-1$ layers out of which only the first and the last is adjacent to $G$. Under these assumptions, we can still assume that $|E''| = O(m) k^3\ell^2 m^{\epsilon^2/k}$ implying that the efficiency of $|E''|$ is at least 
    \begin{align*}
        \eta \coloneqq \frac{\Omega(m) \cdot \Delta^{k-1 - 3\epsilon}}{k^3 \ell^2 m^{1 + \epsilon^2/k}} &= \Omega(\Delta^{k-1-8\epsilon}) m^{-\epsilon^2/k}\\ 
        &= \Omega(\Delta^{k-1-9\epsilon}),
    \end{align*}
    using that $k = O(r), \ell = O(r^3)$, $r = O(\Delta^{2\epsilon})$, and $m = O(r^{7k/\epsilon})$.
   The diameter associated to $E''$ is at most $k m^{\epsilon^2/k} = kO(r^{7\epsilon}) < k\ell$, so, for every critical path, there is still a critical shortest path that uses at least one edge in $E''$. 

    In fact, for each critical shortest path $P$, there is a component $C$ of $E''$ such that $P$ is $C$-respecting while $C$ shortcuts two layers $L_i, L_j$ with $|i - j| \ge 4rk$. Otherwise, the diameter we would achieve for the critical paths would be at least $\frac{\ell k}{4rk} = \Omega(r^2)$, which contradicts our assumption that the achieved diameter is $k m^{\epsilon^2/k} = O(r^{1 + 7\epsilon})$. We therefore delete all components from $E''$ that introduce a shortcut between layers $L_i, L_j$ with $|i - j| < 4rk$ and call the result $E'''$. By the preceding discussion, we still know that every critical shortest path uses at least one edge in $E'''$. Due to this fact and since $|E'''| \le |E''|$, the efficiency of $E'''$ does not decrease as compared to $E''$. 

    We define $\epsilon' \coloneqq 9\epsilon$ and delete all Steiner vertices with degree at least $\Delta^{1+\epsilon'}$ from $E'''$ and call the resulting shortcut $E^*$. According to \Cref{lem:volumnbound} the efficiency of such vertices is at most $\Delta^{k-1-\epsilon'}$, so according to \Cref{lem:efficiencynotdecrease} removing them does not decrease the efficiency of the shortcut. 
    We recall that each component of $E^*$ consists of $t-1$ consecutive layers $S_1, \ldots, S_{t-1}$ such that every $C$-respecting critical shortest path contains exactly one vertex in each layer $S_i$. Denote by $M$ the union of all Steiner vertices in the middle layers $S_{\lfloor t/2 \rfloor}$ across all components $C$ in $E^*$. 
    
    We claim that there exists at least one vertex in $M$ which also has efficiency at least $\eta$. To see this, denote by $E_M$ the set of edges incident to the vertices in $M$ and by $\cP$ the set of $E^*$-respecting critical shortest paths. By definition of efficiency, we know that $|\cP|/|E^*| \geq \eta$, and since $|E_M| \le |E^*|$, we get that $|\cP|/|E_M| \ge \eta$. Now, denote the efficiency of a vertex $v \in M$ by $\eta_v$ and assume for the sake of contradiction that $\eta_v < \eta$ for all $v \in M$. Then, noting that $|\cP| \le \sum_{v\in M} \eta_v \deg(v)$ and $|E_M| = \sum_{v\in M} \deg(v)$ allows us to infer that \begin{align*}
        \frac{|\cP|}{|E_M|} \le \frac{\sum_{v\in M} \eta_v \deg(v)}{\sum_{v\in M} \deg(v)} < \eta,
    \end{align*} which contradicts the fact that $\cP/|E_M| \ge \eta$. 

    We have thus proved that there is some $v \in M$ with efficiency at least $\eta = \Omega(\Delta^{k - 1 - 9\epsilon})$. Denote the number of $E^*$-respecting shortest paths passing through $v$ by $p(v)$ and note that $p(v) \ge \eta \deg(v)$ by definition of efficiency. We show that $p(v)$ is large. To this end, we note that $p(v) \le \deg(v)^2 \Delta^{(1+\epsilon')(t-2)}$. The upper bound is the maximum number of pairs of vertices in $G$, one at the starting layer of $C$ and one at the final layer, both reachable from $v$ using only edges in the component $C$ of $E^*$ containing $v$. This bound uses our upper bound on the degree of remaining Steiner vertices. By uniqueness of critical paths (the third property of~\Cref{lem:criticalpairbasic}), the number of critical shortest paths passing through $v$ is at most as large as this number of such pairs. Using further that $\deg(v) \le p(v)/\eta$ then implies that \begin{align*}
        p(v) &\le \deg(v)^2 \Delta^{(1 + \epsilon')(t-2)} \le \frac{p(v)^2}{\eta^2} \Delta^{(1 + \epsilon')(t-2)},
    \end{align*} which implies that \begin{align*}
        p(v) &\ge \frac{\eta^2}{\Delta^{(t-2)(1 + \epsilon')}} = \Omega(\Delta^{2k - 2 - (t-2)(1 + \epsilon') - 18\epsilon}) = \Omega(\Delta^{2k -t(1+3\epsilon')}).
    \end{align*}
    On the other hand, we know that the critical $E^*$-respecting shortest paths passing through $v$ correspond to critical paths passing through some $A \subseteq L_i$ and some $B \subseteq L_j$ where $L_i, L_j$ are the layers shortcut by the component $C$ of $E^*$ that $v$ belongs to. Due to the maximum degree restriction in $E^*$, and the fact that $v \in M$, we know that $\sigma :=  \max\{|A|, |B|\} \le \Delta^{(\frac{t}{2} + 1)(1+\epsilon')}$. Applying \Cref{lem:efficiencybound}, we get that the number of critical paths with one vertex in $A$ and the other in $B$ is $O((k\ell \sigma \log (n))^{3/2}) = O( \Delta^{\frac{3t}{4}(1 + \epsilon') + \epsilon' + 6\epsilon +  1}) = O( \Delta^{\frac{3t}{4}(1 + 2\epsilon')})$ since $\ell = O(r^3)$ and $r = O(\Delta^{2\epsilon})$. In total, we get that the following inequality must be true \begin{align*}
        \Delta^{2k -t(1+3\epsilon')} \le p(v) \le \Delta^{\frac{3t}{4}(1 + 2\epsilon')}.
    \end{align*}
    This implies that 
    \begin{align*} 2k -t(1+3\epsilon') &\le \frac{3t}{4}(1 + 2\epsilon').
    \end{align*} 
    Rearranging yields that 
    \begin{align*}
        \frac{7}{4} (1 + 3\epsilon') t \ge 2k \Leftrightarrow (1 + 27\epsilon) t \le \frac{8}{7}k.
    \end{align*} Using $\epsilon = 0.0001$ and assuming that \begin{align*}
        1.01t \le \frac{8}{7}k
    \end{align*} would imply that the above inequality is violated. Hence, we have proved that no Steiner shortcut of thickness $t$, diameter $m^{\epsilon^2/k}$, and with $m^{1 + \epsilon^2/k}$ edges can exist if $1.01\cdot t \le \frac{8}{7}k$ and $\epsilon = 0.0001$.\qedhere

\end{proof}

\section{Generalized Shortcuts for Congestion and Length}
\label{sec:generalized}

Let us refer to shortcuts in previous sections as \emph{reachability shortcut}s. In this section, we formalize the generalization of reachability shortcut that also preserves congestion and length in \Cref{subsec:formal}. Then, we explain why the construction of these shortcuts would imply algorithmic breakthroughs in parallel algorithms for shortest paths and maximum flow in \Cref{subsec:implication}. For convenience, we assume lengths and capacities are polynomially bounded integers in this section.

\subsection{Formalization}

\label{subsec:formal}

Here, we define a notion of \emph{length-constrained flow shortcuts} (LC-flow shortcuts). This is a shortcut added to the graph so that every multi-commodity flow can be routed with approximately the same congestion and length, but all flow paths only require a small number of edges. 

To formalize this, we need notation for multi-commodity flows. Let $G=(V,E)$ be a graph with edge capacity $c:E\rightarrow\mathbb{R}_{\ge0}$ and edge length $\ell:E\rightarrow\mathbb{R}_{\ge0}$. A \emph{(multi-commodity) flow} $F$ in $G$ assigns value $F(P)\ge0$ for each directed path $P$ in $G$. We say that $P$ is a flow path of $F$ if $F(P)>0$. The \emph{congestion} on edge $e$ of $F$ is $\congest_{F}(e)=\frac{\sum_{e\in P}F(P)}{c(e)}$ and the congestion of $F$ is $\congest_{F}=\max_{e\in E}\congest_{F}(e)$. The \emph{length} of $F$ is the maximum total length overall flow paths of $F$, i.e., $\max_{P:F(P)>0}\ell(P)$. The \emph{hopbound} of $F$ is the maximum number of edges overall flow paths of $F$, i.e., $\max_{P:F(P)>0}|E(P)|$. For any vertex set $U$, a \emph{demand $D$ on $U$} assigns a value $D(u,v)\ge0$ for each vertex pair $(u,v)\in U$. The demand $D_{F}$ routed by $F$ is defined as $D_{F}(u,v):=\sum_{(u,v)\text{-path }P}F(P)$. We say that $D$ is routable with length $\ell$, congestion $\kappa$, and hopbound $h$ if there exists a flow $F$ with length $\ell$, congestion $\kappa$, and hopbound $h$ that routes $D$. When $D$ is routable in $G$ with congestion $1$, we sometimes just say that $D$ is routable in $G$. With this, we are ready to define LC-flow shortcuts.

\begin{defn}
[Lengh-Constrained Flow Shortcuts]\label{def:flow shortcut}
A shortcut $G'=(V\cup V', E\cup E')$ is a \emph{length-constrained flow shortcut of $G$} with hopbound $h$,\footnote{The hopbound in this section has the same purpose as the diameter in previous sections. We change to our terminology to hopbound to avoid confusion, because the graph here has edge lengths.} congestion slack $\kappa$, and length slack $s$ if:
\begin{enumerate}
\item Every demand on $V$ routable in $G$ with length $\lambda$ and congestion $1$ is routable in $G'$ with length $\lambda s$, congestion $1$, and hopbound $h$, and 
\item Every demand on $V$ routable in $G'$ with length $\lambda$ and congestion $1$ is routable in $G$ with length $\lambda$ and congestion $\kappa$.
\end{enumerate}
\end{defn}

The concept of LC-flow shortcut above was recently defined in undirected graphs \cite{haeupler2024low}. They show that, for any $n$-vertex undirected graph $G$ with capacity and length at most $\poly(n)$ and a parameter $k$, there exists a flow shortcut $G'$ with hopbound $h=O(k)$, length slack $s=O(k)$, congestion slack $\poly(k\log n)n^{O(1/\sqrt{k})}$, and size $|E'|\le n{}^{1+O(1/\sqrt{k})}\poly(\log n)$. It is an interesting open problem if an LC-flow shortcut with non-trivial guarantee exists.

\paragraph{Reachability Shortcuts, Distance Shortcuts, and Congestion Shortcuts.}

Observe that a reachability shortcut is precisely an LC-flow shortcut when the constraints related to length and congestion in \Cref{def:flow shortcut} are removed. %

{} 

We will use the terms \emph{distance shortcut} and \emph{congestion shortcut} to refer to an LC-flow shortcut when the respective constraint related to congestion and length in the definition of \Cref{def:flow shortcut} is removed.

More specifically, when focusing only on length, \Cref{def:flow shortcut} simplifies to distance shortcuts as follows. Let $\dist_{G}(s,t)$ denote the distance from $s$ to $t$ in $G$. Let $\dist_{G}^{h}(s,t)$ denote the minimum length of paths with hopbound $h$ from $s$ to $t$ in $G$. 

\begin{defn}[Distance Shortcuts]\label{def:distanceshortcut}
A shortcut $G'=(V\cup V',E\cup E')$ is a \emph{distance shortcut of $G$} with hopbound $h$ and length slack $s$ if, 
for every vertex pair $(s,t)\in V\times V$, 
\begin{enumerate}
\item If $\dist_{G}(s,t)\le\lambda$, then $\dist_{G'}^{h}(s,t)\le\lambda s$ and 
\item If $\dist_{G'}(s,t)\le\lambda$, then $\dist_{G}(s,t)\le\lambda$.
\end{enumerate}
\end{defn}

When focusing only on congestion, \Cref{def:flow shortcut} simplifies to congestion shortcuts as follows:
\begin{defn}[Congestion Shortcuts]\label{def:congestionshortcut}
A shortcut $G=(V\cup V', E\cup E')$ is a \emph{congestion shortcut of $G$} with hop bound $h$ and congestion slack $\kappa$ if 
\begin{enumerate}
\item Every demand on $V$ routable in $G$ is routable in $G'$ with hopbound $h$, and 
\item Every demand on $V$ routable in $G'$ is routable in $G$ with congestion $\kappa$.
\end{enumerate}
\end{defn}

\subsection{Exact Shortest Paths via $(1+\epsilon)$-Approximate Distance Shortcuts}

\label{subsec:implication}

In this section, we observe that a near-optimal parallel algorithm for constructing a distance shortcut with length slack $(1+o(1/\log n))$ would imply a near-optimal parallel algorithm for \emph{exact} single-source shortest paths with $\tilde{O}(m)$ work and $\polylog(n)$ depth.

In the SSSP problem, given a graph $G$ and a vertex $s$, we need to output the exact distance $d(u)$ from $s$ to every other vertex $u$ in $G$. For $\alpha$-approximate SSSP, we are given a graph $G$ and a vertex $s$, output an approximate distance $\tilde{d}(u)$ from $s$ to every other vertex $u$ in $G$, such that $d(u)\le\tilde{d}(u)\le\alpha d(u)$. 

The reduction from parallel exact SSSP to distance shortcut construction follows from combining tools from \cite{KleinS97} and \cite{RozhonHMGZ23}. 
\begin{thm}
\label{thm:reduc sssp}For any $n,m,m',h$, suppose there is an algorithm $\cA$ which takes an $O(m)$-edge $O(n)$-vertex directed graph and computes a distance shortcut of hopbound $O(h)$, length slack $(1+o(1/\log n))$ and size $O(m')$. Then, there is an algorithm computing exact SSSP of an $m$-edges $n$-vertex directed graph with $\tO 1$ calls to $\cA$, with additional $\tO{m+m'}$ work and $\tO h$ depth. 
\end{thm}

In undirected graphs, there exist distance shortcuts of hopbound $n^{o(1)}$, length slack $(1+1/\log^{2}n)$ and size $n^{1+o(1)}$ (e.g. \cite{cohen2000polylog,Huang2019-thorupzwick}). Moreover, \cite{elkin2019linear} showed how to construct them in $m^{1+o(1)}$ work and $n^{o(1)}$ depth. 

If this can be generalized to directed graphs even when the shortcut size is $m^{1+o(1)}$, then \Cref{thm:reduc sssp} would imply an exact SSSP algorithm with $m^{1+o(1)}$ work and $n^{o(1)}$ depth. 

\begin{cor}
Suppose there is an algorithm that, given $O(m)$-edge $O(n)$-vertex directed graph,  computes a distance shortcut of hopbound $n^{o(1)}$, length slack $(1+o(1/\log n))$, and size $m^{1+o(1)}$. Then, there is an algorithm for computing exact SSSP in an $m$-edge $n$-vertex directed graphs in $m^{1+o(1)}$ work and $n^{o(1)}$ depth.
\end{cor}

Indeed, in directed graphs, the state-of-the-art is still far from this goal, because the problem is even harder than reachability shortcut. There exist distance shortcuts of hopbound $\tilde{O}(n^{1/3})$, length slack $(1+\eps)$ and size $\tilde{O}(n/\eps^{2})$ \cite{bernstein2023closing}, but they require large polynomial time to construct. \cite{cao2020efficient} showed a parallel algorithm with $\tilde{O}(m/\eps^{2})$ work and $n^{0.5+o(1)}$ depth to construct a distance shortcut of hopbound $n^{1/2+o(1)}$, length slack $(1+\eps)$ and size $\tilde{O}(n/\eps^{2})$. 

Note that, since \Cref{thm:reduc sssp} is a \emph{work-efficient }reduction, any improvement to to the $n^{1/2+o(1)}$ hopbound of \cite{cao2020efficient} would immediately improve the start-of-the-art for parallel exact SSSP problem which currently requires $\tilde{O}(m)$ work and $n^{1/2+o(1)}$ depth \cite{cao2020efficient,RozhonHMGZ23}.

\subsubsection*{Proof of \Cref{thm:reduc sssp}}

\begin{proof}
We first present an algorithm for computing $(1+o(1/\log n))$-approximate SSSP in an $O(m)$-edge, $O(n)$-vertex directed graph with $1$ call to $\cA$ and additional $\tO{m+m'}$ time and 
$\tO h$ depth.  We then show how to solve exact SSSP of an $m$-edge $n$-vertex graph using $\tO{1}$ calls to the approximation algorithm with additional $\tO m$ work and $\tO h$ depth. Combining the two reductions finishes the proof.

\paragraph{Approximate SSSP.}
Given a directed $G$ with $O(m)$ edges and $O(n)$ vertices, we first compute a distance shortcut of size $O(m')$, denoted by $G'=(V\cup V',E\cup E')$. 
We will apply \cite[Lemma 3.2]{KleinS97}, 
to $G'$, which we state as follows.

\begin{lemma}[{\cite[Lemma 3.2]{KleinS97}}]\label{lem:lowstepsssp}
    Given a directed graph $G$, a vertex $u$ and parameters $h,\eps$, there is an algorithm computing a \emph{$h$-limited $(1+\eps)$-approximate shortest path} $\tilde{d}_{h}$ defined as $\forall v\in V,d_h(v)\le \tilde{d}_h(v)\le (1+\eps)\cdot d_h(v)$ where $d_h(v)$ is the shortest path distance restricted to only using $h$ edges, in $\tO{mk\eps^{-1}}$ work and $\tO{k\eps^{-1}}$ depth.
\end{lemma}

We can compute an $h$-limited $(1+o(1/\log n))$-approximate shortest path distances $\{\tilde{d}^{G'}(u)\}_{u\in V}$ in $\tO{m+m'}$ work and $\tO h$ depth (by setting the parameter $p=m/k$ in their paper). We use $\{\tilde{d}^{G'}(u)\}_{u\in V}$ as the output. This finishes the algorithm description, which uses $1$ call to $\cA$ with additional $\tO m$ work and $\tO 1$ depth.

In their definition, a $h$-limited $(1+o(1/\log n)$-approximate shortest path distances $\{\tilde{d}^{G'}_h(u)\}_{u\in V}$ satisfies that $d_{h}^{G'}(u)\le\tilde{d}^{G'}_h(u)\le(1+o(1/\log n))d_{h}^{G'}(u)$, where $d_{h}^{G'}(u)$ is the shortest path distance restricted to only using $h$ edges. According to~\cref{def:distanceshortcut}, if we write $d^{G}(u)$ as the shortest path distance on $G$, then we have $d_{h}^{G'}(u)\le(1+o(1/\log n))\cdot d^{G}(u)$ and $d^{G}(u)\le d_{h}^{G'}(u)$. Thus, we have $d^{G}(u)\le\tilde{d}^{G'}_h(u)\le(1+o(1/\log n))(1+o(1/\log n))d^{G}(u)$, which means $\tilde{d}^{G'}_h(u)$ is a $(1+o(1/\log n))$-approximate distance.

\paragraph{Exact SSSP.} The following lemma is a direct implication of Theorem 1.7 and Lemma 3.1 in \cite{RozhonHMGZ23},

\begin{lemma}[{\cite[Theorem 1.7 and Lemma 3.1]{RozhonHMGZ23}}]\label{lem:exacttoapproxsssp}
    Given a directed graph $G$, a source $s\in V(G)$, there is an algorithm computing the exact distance from $s$ to every node in $G$, using $\tO{1}$ calls to a $(1 + O(1/ \log n))$-approximate distance oracle $\cO$ on directed graphs and an additional $\tO{m}$ work and $\tO{1}$ depth.\footnote{One might notice that Lemma 4.1 in \cite{KleinS97} is also about reducing exact shortcut path to approximate shortest path oracles. However, they require a stronger approximate shortest path oracle: The approximate distance must be an exact distant on the shortcut $G'$. Thus, it cannot be used here.}
\end{lemma}

Given an $m$ edges $n$ vertices graph $G$, we apply \Cref{lem:exacttoapproxsssp} to get an exact distance function by $O(\log n)$ calls to the $(1+o(1/\log n))$-approximate distance oracle presented in the last paragraph, and an additional $\tO{m}$ work and $\tO{1}$ depth. The lemma follows.
\end{proof}

\subsection{Parallel Exact Max Flow via $\tilde{O(}1)$-Approximate Congestion Shortcuts}

In this section, we observe that a near-optimal parallel algorithm for constructing a congestion shortcut with congestion slack $\polylog(n)$ (together with an embedding algorithm defined below) would imply a near-optimal parallel algorithm for \emph{exact} maximum flow with $\tilde{O}(m)$ work and $\polylog(n)$ depth.

\begin{defn}
[Backward Mapping]\label{def:embedding} For a shortcut algorithm $\cA$ which outputs a shortcut $G'$ of the input graph $G=(V,E)$, a \emph{backward mapping algorithm} for $\cA$ takes input $G$ and $G'$ generated by $\cA$, along with a flow $F'$ in $G'$ 
with congestion $1$, 
outputs a flow $F$ in $G$ routing the 
same demand as $F'$ with congestion $\kappa$. 
\end{defn}

The reduction from parallel exact max flow to congestion shortcut construction follows quite straightforwardly by repeatedly solving approximate max flow on the residual graph.
\begin{thm}
\label{thm:reduc maxflow}Given $n,m,m',h,\kappa$, suppose there is an algorithm $\cA$ which takes an $O(m)$-edge $O(n)$-vertex directed graph and computes a congestion shortcut of hopbound $O(h)$, congestion slackness $\kappa$ and size $O(m')$, along with a corresponding embedding algorithm $\cA'$ for $\cA$, then there is an algorithm computing exact max flow with $\tO{\kappa}$ calls to $\cA$ and $\cA'$ and additional $\tO{\kappa\cdot\poly(h)\cdot(m+m')}$ work and $\tO{\kappa\cdot\poly(h)}$ depth. 
\end{thm}

In undirected graphs, congestion shortcuts with near-optimal quality follows from \emph{tree flow sparsifiers }defined as follows.
\begin{defn}
A \emph{tree flow sparsifier} $T$ of \emph{$G$} with congestion slack $\kappa$ is such that 
\begin{enumerate}
\item The set of leaves of $T$ corresponds to $V(G)$. 
\item Every demand on $V(G)$ routable in $G$ is routable in $T$ 
\item Every demand on $V(G)$ routable in $T$ is routable in $G$ with congestion $\kappa$. 
\end{enumerate}
\end{defn}

Clearly from the definition, a tree flow sparsifier with depth $d$ and congestion slack $\kappa$ is a congestion shortcut with congestion slack $\kappa$, hopbound $2d$, and size $O(n)$. There exists a tree flow sparsifier with depth $O(\log n)$ and congestion slack $O(\log^{2}n\log\log n)$ \cite{harrelson2003polynomial}. Moreover, there is a parallel algorithm that construct a tree flow sparsifier with depth $O(\log n)$ and congestion slack $O(\log^{9}n)$ in $\tilde{O}(m)$ work and $\polylog(n)$ depth \cite{agarwal2024parallel}. 

If this can be generalized to directed graphs even when the shortcut size is $\tilde{O}(m)$, then \Cref{thm:reduc maxflow} would imply an exact max flow algorithm with $\tilde{O}(m)$ work and $\polylog(n)$ depth.

\begin{cor}
Suppose there is an algorithm that, given $O(m)$-edge $O(n)$-vertex directed graph,  computes a congestion shortcut of hopbound $\polylog(n)$, length slack $\polylog(n)$, and size $\tilde{O}(m)$. Then, there is an algorithm for computing exact max flow in an $m$-edge $n$-vertex directed graphs in $\tilde{O}(m)$ work and $\polylog(n)$ depth.
\end{cor}

Unfortunately, there is no prior work on this and we leave it an interesting open problem.

\subsubsection*{Proof of \Cref{thm:reduc maxflow}}

\begin{proof}
We first give an algorithm computing approximate max flow which routes a feasible flow whose value is at least $1/2\kappa$ times the optimal value on the input graph. Then show how to turn it into an exact max flow algorithm.

\paragraph{Approximate max flow.}

We first use $\cA$ to compute a congestion shortcut on $G$, denoted by $E'$. Let $G'=G\cup E'$. We will use Theorem 3.1 (1) of \cite{HaeuplerHS23}, stated as follows.

\begin{lemma}[Theorem 3.1 (1) of \cite{HaeuplerHS23}]\label{lem:lowstepapxflow}
    Given a directed graph $G$ with edge capacities and unit edge length, length bound $h\ge 1$, approximation parameter $\eps>0$, and source and sink $s,t\in V$, there is an algorithm computing a $1-\eps$ approximate maximum $h$-length flow, in the sense that the flow value is at least $1-\eps$ times the maximum $h$-length flow from $s$ to $t$ in $G$. The algorithm takes $\tO{m\cdot h^{17}\cdot \frac{1}{\eps^{9}}}$ work and $\tO{h^{17}\cdot \frac{1}{\eps^{9}}}$ depth.
\end{lemma}

We use \Cref{lem:lowstepapxflow} with parameter $\eps=1/2$, $h$ the same as our $h$, and the digraph the same as $G'$ with unit length. The guarantee for the output is a flow $F'$, which is a $1-\eps$ approximation for maximum $h$-length flow, which, in our case, is the maximum flow with hop bound $h$ since the graph has unit length. Then, we call the backward mapping algorithm $\cA$ to turn $F'$ into a flow $F$ on $G$ with congestion $\kappa$. To make $F$ feasible, we scale the flow on each edge down by a factor of $\kappa$, and denote the final flow as $F/\kappa$. According to~\cref{def:flow shortcut}, every flow on $G$ corresponds to a feasible flow on $G'$ with hop bound $h$ and the same value. Thus, $F'$ has a value of at least $1/2$ times the max flow value on $G$, which implies that $F/\kappa$ is a $1/2\kappa$ approximate max flow on $G$.

The algorithm takes an input graph of $m$ edges and $n$ vertices, uses $1$ call to $\cA$ and $\cA'$, and additional $\tO{\poly(h)\cdot(m+m')}$ work and $\poly(h)$ depth using \Cref{lem:lowstepapxflow}.

\paragraph{Exact max flow.}

If we have an algorithm computing $\eps$-approximate max flow (here $\eps=1/2\kappa$), then we can use a folklore reduction to get the exact max flow: iteratively compute approximate max flows on the residual graph. 
In $\tO{1/\eps}$ iterations the optimal flow value of the residual graph will be decreased to $1/\poly(n)$. Notice that we can also apply the standard rounding algorithm (for example, see~\cite[\S 5]{Cohen95}) 
if we want to get an integer flow of the max flow value.

The algorithm uses $\tO{\kappa}$ calls to the approximate max flow algorithm, which corresponds to $\tO{\kappa}$ calls to $\cA,\cA'$, and additional $\tO{\kappa\cdot\poly(h)\cdot(m+m')}$ work and $\tO{\kappa\cdot\poly(h)}$ depth.
\end{proof}
\section*{Acknowledgments}
AB is supported by Sloan Fellowship, Google Research Fellowship,  NSF Grant 1942010, and Charles S. Baylis endowment at NYU.

HF is supported by the National Science Foundation Graduate
Research Fellowship Program under Grant No.\ DGE2140739. 

BH is partially funded by the Ministry of Education and Science of Bulgaria's support for INSAIT as part of the Bulgarian National Roadmap for Research Infrastructure and through the European Research Council (ERC) under the European Union's Horizon 2020 research and innovation program (ERC grant agreement 949272).

GH is supported in part by National Science Foundation grant 
NSF CCF-2153680.

GL is supported by the National Science Foundation Graduate
Research Fellowship Program under Grant No.\ DGE2140739.

MP is supported by grant no. 200021 204787 of the Swiss National Science Foundation.

SP is supported by NSF Grant CCF-2221980.

TS is supported by Sloan Fellowship, NSF Grant CCF-2238138, and partially funded by the Ministry of Education and Science of Bulgaria's support for INSAIT, Sofia University ``St.~Kliment Ohridski'' as part of the Bulgarian National Roadmap for Research Infrastructure.

\bibliographystyle{alpha}
\bibliography{refs}

\newcommand{\etalchar}[1]{$^{#1}$}
\begin{thebibliography}{RHM{\etalchar{+}}23b}

\bibitem[AB24]{abboud2024reachability}
Amir Abboud and Greg Bodwin.
\newblock Reachability preservers: New extremal bounds and approximation algorithms.
\newblock {\em SIAM J. Comput.}, 53(2):221--246, 2024.

\bibitem[AKL{\etalchar{+}}24]{agarwal2024parallel}
Arpit Agarwal, Sanjeev Khanna, Huan Li, Prathamesh Patil, Chen Wang, Nathan White, and Peilin Zhong.
\newblock Parallel approximate maximum flows in near-linear work and polylogarithmic depth.
\newblock In {\em Proceedings of the 35th Annual ACM-SIAM Symposium on Discrete Algorithms (SODA)}, pages 3997--4061, 2024.

\bibitem[BBG{\etalchar{+}}14]{BermanBGRWY14}
Piotr Berman, Arnab Bhattacharyya, Elena Grigorescu, Sofya Raskhodnikova, David~P. Woodruff, and Grigory Yaroslavtsev.
\newblock Steiner transitive-closure spanners of low-dimensional posets.
\newblock {\em Comb.}, 34(3):255--277, 2014.

\bibitem[BH23a]{bodwin2023folklore}
Greg Bodwin and Gary Hoppenworth.
\newblock Folklore sampling is optimal for exact hopsets: Confirming the sqrt(n) barrier.
\newblock In {\em 2023 IEEE 64th Annual Symposium on Foundations of Computer Science (FOCS)}, pages 701--720, 2023.

\bibitem[BH23b]{BodwinH23}
Greg Bodwin and Gary Hoppenworth.
\newblock Folklore sampling is optimal for exact hopsets: Confirming the $\sqrt{n}$ barrier.
\newblock In {\em Proceedings of the 64th Annual IEEE Symposium on Foundations of Computer Science ({FOCS})}, pages 701--720, 2023.

\bibitem[BL98]{barany1998convex}
Imre B{\'a}r{\'a}ny and David~G Larman.
\newblock The convex hull of the integer points in a large ball.
\newblock {\em Mathematische Annalen}, 312:167--181, 1998.

\bibitem[BW23]{bernstein2023closing}
Aaron Bernstein and Nicole Wein.
\newblock Closing the gap between directed hopsets and shortcut sets.
\newblock In {\em Proceedings of the 2023 Annual ACM-SIAM Symposium on Discrete Algorithms (SODA)}, pages 163--182, 2023.

\bibitem[CFR20]{cao2020efficient}
Nairen Cao, Jeremy~T Fineman, and Katina Russell.
\newblock Efficient construction of directed hopsets and parallel approximate shortest paths.
\newblock In {\em Proceedings of the 52nd Annual ACM SIGACT Symposium on Theory of Computing (STOC)}, pages 336--349, 2020.

\bibitem[Cha87]{Chazelle87}
Bernard Chazelle.
\newblock Computing on a free tree via complexity-preserving mappings.
\newblock {\em Algorithmica}, 2:337--361, 1987.

\bibitem[CKL{\etalchar{+}}22]{chen2022maximum}
Li~Chen, Rasmus Kyng, Yang~P Liu, Richard Peng, Maximilian~Probst Gutenberg, and Sushant Sachdeva.
\newblock Maximum flow and minimum-cost flow in almost-linear time.
\newblock In {\em Proceedings of the 2022 63rd Annual IEEE Symposium on Foundations of Computer Science (FOCS)}, pages 612--623, 2022.

\bibitem[Coh95]{Cohen95}
Edith Cohen.
\newblock Approximate max-flow on small depth networks.
\newblock {\em {SIAM} J. Comput.}, 24(3):579--597, 1995.

\bibitem[Coh97]{cohen1997size}
Edith Cohen.
\newblock Size-estimation framework with applications to transitive closure and reachability.
\newblock {\em Journal of Computer and System Sciences}, 55(3):441--453, 1997.

\bibitem[Coh00]{cohen2000polylog}
Edith Cohen.
\newblock Polylog-time and near-linear work approximation scheme for undirected shortest paths.
\newblock {\em Journal of the ACM (JACM)}, 47(1):132--166, 2000.

\bibitem[DRNV16]{de2016limited}
Susanna~F De~Rezende, Jakob Nordstr{\"o}m, and Marc Vinyals.
\newblock How limited interaction hinders real communication (and what it means for proof and circuit complexity).
\newblock In {\em Proceedings of the 57th Annual IEEE Symposium on Foundations of Computer Science (FOCS)}, pages 295--304, 2016.

\bibitem[EN19]{elkin2019linear}
Michael Elkin and Ofer Neiman.
\newblock Linear-size hopsets with small hopbound, and constant-hopbound hopsets in {RNC}.
\newblock In {\em Proceedings of the 31st ACM Symposium on Parallelism in Algorithms and Architectures (SPAA)}, pages 333--341, 2019.

\bibitem[Fin20]{Fineman20}
Jeremy~T. Fineman.
\newblock Nearly work-efficient parallel algorithm for digraph reachability.
\newblock {\em {SIAM} J. Comput.}, 49(5), 2020.

\bibitem[GKW21]{GolovnevKW21}
Alexander Golovnev, Alexander~S. Kulikov, and R.~Ryan Williams.
\newblock Circuit depth reductions.
\newblock In {\em Proceedings of the 12th Innovations in Theoretical Computer Science Conference ({ITCS})}, volume 185 of {\em LIPIcs}, pages 24:1--24:20. Schloss Dagstuhl - Leibniz-Zentrum f{\"{u}}r Informatik, 2021.

\bibitem[H{\aa}s86]{Hastad86}
Johan H{\aa}stad.
\newblock Almost optimal lower bounds for small depth circuits.
\newblock In {\em Proceedings of the 18th Annual {ACM} Symposium on Theory of Computing (STOC)}, pages 6--20, 1986.

\bibitem[Hes03]{Hesse03}
William Hesse.
\newblock Directed graphs requiring large numbers of shortcuts.
\newblock In {\em Proceedings of the 14th Annual ACM-SIAM Symposium on Discrete Algorithms (SODA)}, pages 665--669, 2003.

\bibitem[HHL{\etalchar{+}}24]{haeupler2024low}
Bernhard Haeupler, D~Ellis Hershkowitz, Jason Li, Antti Roeyskoe, and Thatchaphol Saranurak.
\newblock Low-step multi-commodity flow emulators.
\newblock In {\em Proceedings of the 56th Annual ACM Symposium on Theory of Computing (STOC)}, pages 71--82, 2024.

\bibitem[HHR03]{harrelson2003polynomial}
Chris Harrelson, Kirsten Hildrum, and Satish Rao.
\newblock A polynomial-time tree decomposition to minimize congestion.
\newblock In {\em Proceedings of the 15th Annual ACM Symposium on Parallel Algorithms and Architectures (SPAA)}, pages 34--43, 2003.

\bibitem[HHS23]{HaeuplerHS23}
Bernhard Haeupler, D.~Ellis Hershkowitz, and Thatchaphol Saranurak.
\newblock Maximum length-constrained flows and disjoint paths: Distributed, deterministic, and fast.
\newblock In {\em Proceedings of the 55th Annual ACM Symposium on Theory of Computing ({STOC})}, pages 1371--1383, 2023.

\bibitem[HP19]{Huang2019-thorupzwick}
Shang-En Huang and Seth Pettie.
\newblock Thorup-{Z}wick emulators are universally optimal hopsets.
\newblock {\em Information Processing Letters}, 142:9--13, 2019.

\bibitem[HP21]{HuangP21}
Shang{-}En Huang and Seth Pettie.
\newblock Lower bounds on sparse spanners, emulators, and diameter-reducing shortcuts.
\newblock {\em {SIAM} J. Discret. Math.}, 35(3):2129--2144, 2021.

\bibitem[HXX25]{hopp25}
Gary Hoppenworth, Yinzhan Xu, and Zixuan Xu.
\newblock New separations and reductions for directed hopsets and preservers.
\newblock In {\em Proceedings of the 2025 Annual ACM-SIAM Symposium on Discrete Algorithms (SODA)}, 2025.

\bibitem[JLS19]{LiuJS19}
Arun Jambulapati, Yang~P. Liu, and Aaron Sidford.
\newblock Parallel reachability in almost linear work and square root depth.
\newblock In {\em Proceedings of the 60th Annual {IEEE} Symposium on Foundations of Computer Science ({FOCS})}, pages 1664--1686, 2019.

\bibitem[KK93]{KaoK93}
Ming{-}Yang Kao and Philip~N. Klein.
\newblock Towards overcoming the transitive-closure bottleneck: Efficient parallel algorithms for planar digraphs.
\newblock {\em J. Comput. Syst. Sci.}, 47(3):459--500, 1993.

\bibitem[KP22a]{kogan2022beating}
Shimon Kogan and Merav Parter.
\newblock Beating matrix multiplication for $n^{1/3}$-directed shortcuts.
\newblock In {\em Proceedings of the 49th International Colloquium on Automata, Languages, and Programming (ICALP)}. Schloss-Dagstuhl-Leibniz Zentrum f{\"u}r Informatik, 2022.

\bibitem[KP22b]{kogan2022new}
Shimon Kogan and Merav Parter.
\newblock New diameter-reducing shortcuts and directed hopsets: Breaking the $o(\sqrt{n})$ barrier.
\newblock In {\em Proceedings of the 33rd Annual ACM-SIAM Symposium on Discrete Algorithms (SODA)}, pages 1326--1341, 2022.

\bibitem[KP23]{kogan2023faster}
Shimon Kogan and Merav Parter.
\newblock Faster and unified algorithms for diameter reducing shortcuts and minimum chain covers.
\newblock In {\em Proceedings of the 2023 Annual ACM-SIAM Symposium on Discrete Algorithms (SODA)}, pages 212--239, 2023.

\bibitem[KS97]{KleinS97}
Philip~N. Klein and Sairam Subramanian.
\newblock A randomized parallel algorithm for single-source shortest paths.
\newblock {\em J. Algorithms}, 25(2):205--220, 1997.

\bibitem[LVWX22]{LuWWX22}
Kevin Lu, Virginia {Vassilevska Williams}, Nicole Wein, and Zixuan Xu.
\newblock Better lower bounds for shortcut sets and additive spanners via an improved alternation product.
\newblock In {\em Proceedings of the 2022 {ACM-SIAM} Symposium on Discrete Algorithms, ({SODA})}, pages 3311--3331, 2022.

\bibitem[Ras10]{Raskhodnikova10}
Sofya Raskhodnikova.
\newblock Transitive-closure spanners: {A} survey.
\newblock In Oded Goldreich, editor, {\em Property Testing --- Current Research and Surveys}, volume 6390 of {\em Lecture Notes in Computer Science}, pages 167--196. Springer, 2010.

\bibitem[RHM{\etalchar{+}}23a]{rozhovn2023parallel}
V{\'a}clav Rozho{\v{n}}, Bernhard Haeupler, Anders Martinsson, Christoph Grunau, and Goran Zuzic.
\newblock Parallel breadth-first search and exact shortest paths and stronger notions for approximate distances.
\newblock In {\em Proceedings of the 55th Annual ACM Symposium on Theory of Computing (STOC)}, pages 321--334, 2023.

\bibitem[RHM{\etalchar{+}}23b]{RozhonHMGZ23}
V{\'{a}}clav Rozhon, Bernhard Haeupler, Anders Martinsson, Christoph Grunau, and Goran Zuzic.
\newblock Parallel breadth-first search and exact shortest paths and stronger notions for approximate distances.
\newblock In {\em Proceedings of the 55th Annual ACM Symposium on Theory of Computing ({STOC})}, pages 321--334. {ACM}, 2023.

\bibitem[RM97]{raz1997separation}
Ran Raz and Pierre McKenzie.
\newblock Separation of the monotone nc hierarchy.
\newblock In {\em Proceedings 38th Annual Symposium on Foundations of Computer Science (FOCS)}, pages 234--243, 1997.

\bibitem[Sch82]{Schnitger82}
Georg Schnitger.
\newblock A family of graphs with expensive depth reduction.
\newblock {\em Theor. Comput. Sci.}, 18:89--93, 1982.

\bibitem[Sch83]{Schnitger83}
Georg Schnitger.
\newblock On depth-reduction and grates.
\newblock In {\em Proceedings of the 24th Annual IEEE Symposium on Foundations of Computer Science (FOCS)}, pages 323--328, 1983.

\bibitem[Tho92]{thorup1993shortcutting}
Mikkel Thorup.
\newblock On shortcutting digraphs.
\newblock In {\em Proceedings of the 18th International Workshop on Graph-Theoretic Concepts in Computer Science (WG)}, volume 657 of {\em Lecture Notes in Computer Science}, pages 205--211. Springer, 1992.

\bibitem[Tho95]{thorup1995shortcutting}
Mikkel Thorup.
\newblock Shortcutting planar digraphs.
\newblock {\em Combinatorics, Probability and Computing}, 4(3):287--315, 1995.

\bibitem[Tho04]{Thorup04}
Mikkel Thorup.
\newblock Compact oracles for reachability and approximate distances in planar digraphs.
\newblock {\em J. {ACM}}, 51(6):993--1024, 2004.

\bibitem[UY91]{ullman1990high}
Jeffrey~D. Ullman and Mihalis Yannakakis.
\newblock High-probability parallel transitive-closure algorithms.
\newblock {\em {SIAM} J. Comput.}, 20(1):100--125, 1991.

\bibitem[VXX24]{WilliamsXX24}
Virginia {Vassilevska Williams}, Yinzhan Xu, and Zixuan Xu.
\newblock Simpler and higher lower bounds for shortcut sets.
\newblock In {\em Proceedings of the 2024 {ACM-SIAM} Symposium on Discrete Algorithms, ({SODA})}, pages 2643--2656, 2024.

\end{thebibliography}

\appendix

\section{Lower Bounds for $m^{1-\epsilon}$-Sized Steiner Shortcuts}
\label{sec:lower bound sublinear m}

In this section, we  present a simple lower bound ruling out the existence of strongly sublinear Steiner reachability shortcuts with subpolynomial diameter. 

\begin{restatable}{theorem}{mepsLB} \label{thm:mepsLB}
For every sufficiently small constant $\eps > 0$, there exists an infinite family of $n$-node, $m$-edge directed graphs $G$, where $m \in \left[n, n^{2-4\eps^{1/2}+\eps}\right]$, such that  every Steiner reachability shortcut of $G$ has size at least  $\tilde{\Omega}(m^{1-\eps})$ or diameter at least $\Omega(n^{\eps})$. 
\end{restatable}

Our lower bound argument will require the following lemma, which can be attributed to existing lower bounds for reachability preservers in~\cite{abboud2024reachability}.

\begin{lemma}[cf. Theorem 13 of \cite{abboud2024reachability}]
\label{lem:rp_lb}
For every sufficiently small $\eps > 0$, there exists an infinite family of $n$-node directed graphs $G$,  each with an associated collection of directed paths $\mathcal{P}$ in $G$ with the following properties:
\begin{enumerate}
    \item $|\mathcal{P}| \geq n^{2-4\eps^{1/2}}$,
    \item each path $P \in \mathcal{P}$ has $|P| = n^{\eps}$ edges, 
    \item paths in $\mathcal{P}$ are pairwise edge-disjoint, and
    \item each path $P \in \mathcal{P}$ is the unique directed path between its endpoints in graph $G$. 
\end{enumerate}
\end{lemma}

Using \Cref{lem:rp_lb}, we can directly prove \Cref{thm:mepsLB}. 
\begin{proof}[Proof of \Cref{thm:mepsLB}]
Let $\eps > 0$ be a sufficiently small constant and let $m$ be some function of $n$ such that $m \in \left[n, n^{2-4\eps^{1/2}+\eps}\right]$. By \Cref{lem:rp_lb}, there exists an infinite family of $n$-node directed graphs $G$, each with an associated collection of directed paths $\mathcal{P}$ in $G$ satisfying the properties in \Cref{lem:rp_lb}. Since paths in $\mathcal{P}$ are pairwise edge-disjoint and each has length at least $n^{\eps}$, we observe that
$$
|E(G)| \geq |\mathcal{P}|n^{\eps} \geq n^{2-4\eps^{1/2}+\eps}.
$$
We can ensure that $|E(G)| = \Theta(m)$ by repeatedly removing paths in $\mathcal{P}$ from graph $G$ until we achieve the desired graph density. After performing this procedure, $|E(G)| = m$ and $|\mathcal{P}| = mn^{-\eps} \geq m^{1-\eps}$, by Property 2 of \Cref{lem:rp_lb}. 

We will achieve our Steiner shortcut lower bound via an information theoretic argument. More specifically, we use graph $G$ to construct a graph family $\mathcal{G}$ of size $|\mathcal{G}| = 2^{|\mathcal{P}|}$. For each subset $S \subseteq \mathcal{P}$, we define graph $G^S$ as follows:
\begin{enumerate}
    \item Initially, let $G^S := G$.
    \item For each path $P \in \mathcal{P} - S$, remove every edge in $P$ from graph $G_S$, so that $G^S \gets G^S - E(P)$. 
\end{enumerate}
Since paths in $\mathcal{P}$ are pairwise edge-disjoint in $G$ by \cref{lem:rp_lb}, a path $P \in \mathcal{P}$ is contained in graph $G^S$ if and only if $P \in S$. 

For each graph $G^S  = (V^S, E^S)$ where $S \subseteq \mathcal{P}$, let $G^S_+ = (V^S \cup V^S_+, E^S \cup E^S_+)$ be a Steiner shortcut graph of $G^S$ with diameter at most $n^{\eps} / 2$ that minimizes $|V^S_+| + |E^S_+|$.  (Recall that the size of a Steiner shortcut graph is precisely $|V^S_+| + |E^S_+|$.) We claim that one of these graphs $G_+^S$ has size at least $|V^S_+| + |E^S_+| \geq \tilde{\Omega}(m^{1-\eps})$.

For each $S \subseteq \mathcal{P}$, let $H^S_+ = G^S_+ - E^S$. This graph contains exactly the Steiner edges $E^S_+$ in graph $G^S_+$. 
We define the graph family $\mathcal{G}$ to be $\mathcal{G} = \{H^S_+ \}_{S \subseteq \mathcal{P}}$.  
Let $x = \max_{S \subseteq \mathcal{P}}|V_+^S| + |E_+^S|$. We observe that each graph  $H_+^S \in \mathcal{G}$ has at most $x+n$ vertices and at most $x$ edges. 
Now note that without loss of generality, we can interpret each directed graph $H^S_+ \in \mathcal{G}$ as a subgraph of a vertex-labeled,  complete, directed graph with $x$ vertices. We denote this universal graph as $\hat{G}$.  Under this interpretation, $H^S_+ \subseteq \hat{G}$ for each $H^S_+ \in \mathcal{G}$.

Now we claim that for distinct sets of paths $S, S' \subseteq \mathcal{P}$, we have that $H^S_+ \neq H^{S'}_+$ (we say that $H^S_+ = H^{S'}_+$ if both graphs correspond to the same subgraph of $\hat{G}$, and we say that $H^S_+ \neq H^{S'}_+$ otherwise).  Since $S \neq S'$, we may assume without loss of generality that there exists a path $P \in S - S'$. Since $G^S_+$ is a Steiner shortcut graph of $G^S$ with diameter $\leq n^{\eps}/2$, and  path $P$ has length $|P| = n^{\eps}$,  and path $P$ is a unique path in $G^S$ by Property 4 of \Cref{lem:rp_lb}, we conclude that there must exist a path in graph $H^S_+ = G^S_+ - E^S$ between a pair of vertices $(s, t)$ in the transitive closure of path $P$. On the other hand, no edge in path $P$ is contained in graph $G^{S'}$, i.e.,  $E(P) \cap E^{S'} = \emptyset$, so by Property 4 of \Cref{lem:rp_lb}, there does not exist a path in graph $G^{S'}_+$ between a pair of vertices $(s, t)$ in the transitive closure of path $P$. We conclude that $G_+^S \neq G_+^{S'}$, as claimed. This implies that there are $2^{|\mathcal{P}|}$ distinct, vertex-labeled graphs in graph family $\mathcal{G}$. 

Finally, we observe that there are at most $(x+n)^{2x}$ distinct subgraphs of the universal graph $\hat{G}$ with at most $x+n$ vertices and at most $x$ edges, so $(x+n)^{2x} \geq 2^{|\mathcal{P}|}$. Taking the log of both sides, we get $x \log (x+n) \geq |\mathcal{P}|$, which implies that $x \geq |\mathcal{P}| / \log(|\mathcal{P}|) = \tilde{\Omega}(m^{1-\eps})$. Then by the definition of $x$, one of the Steiner shortcut graphs $G^S_+$, where $S \subseteq \mathcal{P}$, has size at least $x = \tilde{\Omega}(m^{1-\eps})$. This completes the proof. 
\end{proof}

\section{Reachability Shortcuts as OR-circuit Depth Reduction}
In this section we will show that the existence of Steiner shortcuts for reachability is equivalent to the existence of a depth reduction for a certain class of Boolean circuits. 

\begin{restatable}{theorem}{conj}
    \Cref{conj:Steiner} is equivalent to \Cref{conj:circuit}. \label{thm:equiv_conj}
\end{restatable}

\paragraph{Preliminaries.} A Boolean circuit is a directed acyclic graph with a set of source nodes $S$, called input gates, and a set of sink nodes $T$, called output gates. The input gates are labeled with binary variables $x_1, \dots, x_n$, and the output gates are labeled with binary variables $y_1, \dots, y_n$. 

Each non-source/sink node in the circuit computes a Boolean function (e.g., AND, OR, NOT). We will restrict our attention to circuits with only OR gates, which we refer to as OR-circuits. Our OR gates will have unbounded fan-in. We measure the size of a Boolean circuit by the number of edges it has. We define the depth of a circuit as the length of the longest path in the directed acyclic graph. We define the \textit{diameter} of a circuit as the largest distance between an input gate and an output gate in the directed acyclic graph.

The connection between reachability and OR-circuits can be summarized by the following observation.

\begin{observation}
Let $C$  be an OR-circuit with input gates $S \subseteq V(C)$ and output gates $T \subseteq V(C)$.  
Let $t \in T$ be an output gate labeled with binary variable $y \in \{0, 1\}$. 
Then output gate $t$ returns the value $y = 1$ if and only if there exists an input gate $s \in S$ such that $s$ can reach $t$ in $C$, and input gate $s$ is labeled with a binary variable $x \in \{0, 1\}$ that has the value $x = 1$.
\label{obs:or_reach}
\end{observation}

We can now prove an equivalence between diameter-reducing Steiner shortcuts and the existence of diameter-reductions for OR-circuits.

\begin{lemma}
Let $s(m)$ and $d(m)$ be functions such that $s(m), d(m) = O(\text{\normalfont poly}(m))$. Then the following two statements are equivalent:
  \begin{enumerate}[label=(\arabic*),ref=(\arabic*)]
   \item Every directed graph on $n$ nodes and $m$ edges admits a Steiner reachability shortcut with $s(m) + O(m+n)$ additional edges and diameter $d(m)+O(1)$.
   \item Let  $f:\{0, 1\}^n \mapsto \{0, 1\}^n$ be a multi-output Boolean function  that can be computed by a Boolean \text{\normalfont OR}-circuit with $m$ edges. Then $f$ can be computed by an  \text{\normalfont OR}-circuit with $s(m) + O(m+n)$  edges and diameter $d(m)+O(1)$. 
  \end{enumerate}
    \label{lem:or_circuit}
\end{lemma}
\begin{proof}
We will prove each direction separately.

\paragraph{$(1) \implies (2)$.} Let  $f:\{0, 1\}^n \mapsto \{0, 1\}^n$ be a multi-output Boolean function  that can be computed by a Boolean \text{\normalfont OR}-circuit with $m$ edges. Let $G$ be the DAG of the Boolean circuit computing $f$, where $|E(G)| = m$. Graph $G$ has source nodes $S \subseteq V(G)$ and sink nodes $T \subseteq V(G)$, corresponding to input gates and output gates respectively, such that $|S| = |T| = n$.  

Let $H$ be a Steiner reachability shortcut of graph $G$, with $V(G) \subseteq V(H)$ and $E(G) \subseteq E(H)$. We may assume without loss of generality that $H$ is a DAG, by contracting all strongly connected components (SCCs) in $H$. This operation will not contract any nodes in $S$ or $T$. Additionally, we can also assume wlog that every source node in $H$ is contained in $S$, and every sink node in $H$ is contained in $T$. Consequently, we can interpret $H$ as an OR-circuit with input gates $S$ and output gates $T$. We claim that circuit $H$ computes function $f$ and  satisfies Statement (2).

By Statement (1), $|E(H)| = s(m) + O(m+ n)$ and $H$ has depth $d(m) + O(1)$, as desired. Since $H$ is a Steiner reachability shortcut of $G$, for all $(s, t) \in S \times T$, node $s$ can reach node $t$ in $G$ if and only if node $s$ can reach node $t$ in $H$. Then by Observation \ref{obs:or_reach}, circuit $H$ must compute the same Boolean function as circuit $G$. Then circuit $H$ computes function $f$, as claimed.  

\paragraph{$(2) \implies (1)$.} Let $G$ be an $n$-node, $m$-edge directed graph. We may assume wlog that $m \geq n-1$. Likewise, we may assume wlog that $G$ is a directed acyclic graph, via the following procedure. Contract every SCC in $G$ to obtain a DAG $G'$. Then a simple argument will show that $G$ has a Steiner reachability shortcut with $s(m) + O(m + n)$ additional edges and depth $d(m) + O(1)$ if and only if $G'$ has one as well. 

We will convert DAG $G$ into a circuit $C$ as follows.
Let $v_1, \dots, v_n \in V(G)$ be the nodes of graph $G$. 
We will define sets of new nodes $S = \{s_1, \dots, s_n\}$ and $T = \{t_1, \dots, t_n\}$. Nodes $S$ and $T$ will become the input gates and output gates, respectively, of $C$. Let the vertex set of $C$ be $V(C) = V(G) \cup S \cup T$. We define the edge set $E(C)$ of $C$ to be
$$
E(C) =  E(G) \cup \bigcup_{i \in [1, n]} \{(s_i, v_i), (v_i, t_i) \}.
$$
Notice that $C$ is a directed acyclic graph with source nodes $S$ and sink nodes $T$. Consequently, we can interpret $C$ as an OR-circuit with input gates $S$ and output gates $T$. We label node $s_i$ with binary variable $x_i$, and we label node $t_i$ with binary variable $y_i$, for $i \in [1, n]$. 
Let $f:\{0, 1\}^n \mapsto \{0, 1\}^n$ denote the Boolean function that circuit $C$ computes. 

Let $m' = |E(C)| = m + 2n$. 
By Statement (2), there exists an OR-circuit $C'$ with sources $S$ and sinks $T$, such that
\begin{itemize}
    \item circuit $C'$ computes function $f$,
    \item circuit $C'$ has $s(m') + O(m' + n) = s(m) + O(m + n)$ edges, and 
    \item circuit $C'$ has depth $d(m') + O(1) = d(m) + O(1)$. 
\end{itemize}  
We convert circuit $C'$ into a Steiner reachability shortcut $H$ of graph $G$ as follows. 
Let $V(H)$ be the disjoint union of $V(C')$ and the vertex set $V(G)$ of $G$. We define the edge set $E(H)$ of $H$ to be
$$
E(H) = E(C') \bigcup_{i \in [1, n]} \{ (v_i, s_i), (t_i, v_i) \}.
$$
We claim that graph $H$ is a Steiner reachability shortcut of $G$ and satisfies Statement $(1)$.

We will need the following observation about circuit $C'$. Fix an $i \in [1, n]$, and let $x_i = 1$ and $x_j = 0$ for $j \neq i$. Let $f(x_1, \dots, x_n) = (y_1, \dots, y_n)$. Then by Observation \ref{obs:or_reach}, for any $j \in [1, n]$, we have that $y_j = 1$ if and only if $v_i$ can reach $v_j$ in $G$.  Since circuit $C'$ computes function $f$, we conclude that for all $(s_i, t_j) \in S \times T$, we have that $s_i$ can reach $t_j$ in $C'$ if and only if $v_i$ can reach $v_j$ in graph $G$. 

We will first show that graph $H$ preserves the reachability relation of $G$ between vertices in $V(G) \times V(G)$.  Suppose that node $v_i$ can reach node $v_j$ in graph $G$. Then since we  added edges $(v_i, s_i)$ and $(t_j, v_j)$ to graph $H$, the above observation immediately implies that  $v_i$ can reach $v_j$ in graph $H$. 

Now suppose that $v_i$ can reach $v_j$ in graph $H$; we claim that $v_i$ can reach $v_j$ in graph $G$ as well. Let $P$ be a $v_i \leadsto v_j$ path in graph $H$. We can write path $P$ as a concatenation of paths
$$
P = P_1 \circ \dots \circ P_k, 
$$
where, for each $\ell \in [1, k]$,  subpath $P_{\ell}$ is internally vertex-disjoint from $V(G)$. 
For each $\ell \in [1, k]$, let subpath $P_{\ell}$ be an $x_{\ell} \leadsto x_{\ell+1}$ path, where $x_{\ell}, x_{\ell+1} \in V(G)$,  $x_1 = v_i$, and $x_{k+1} = v_j$.  
Notice that since path $P_{\ell}$ is internally vertex-disjoint from $V(G)$, path $P_{\ell}$ is a path in circuit $C'$ as well, for all $\ell \in [1, k]$. Then by our earlier observation about circuit $C'$, we must have that $x_{\ell}$ can reach $x_{\ell+1}$ in $G$. We have shown that $x_{\ell}$ can reach $x_{\ell+1}$ in $G$ for all $\ell \in [1, k]$, so by the transitivity of the reachability relation, we must have that $x_1 = v_i$ can reach $x_{k+1} = v_j$ in $G$. We have shown that $H$ preserves the reachability relation of $G$ between vertices in $V(G) \times V(G)$. 

Finally, we quickly observe that $|E(H)| = |E(C')| + 2n = s(m) + O(m + n)$, and $H$ has diameter $d(m) + O(1)$, as claimed in Statement (2).
\end{proof}

With \Cref{lem:or_circuit} in hand, we can now prove that our conjecture for  low-depth Steiner shortcuts  (\Cref{conj:Steiner}) is equivalent to the conjecture for low-depth OR-circuits (\Cref{conj:circuit}).

\conj*
\begin{proof}
    First, notice that if a Boolean OR-circuit has depth $d$, then it also has diameter $d$. Therefore, by \Cref{lem:or_circuit}, if \Cref{conj:circuit} is true, then every directed graph on $n$ nodes and $m$ edges admits a Steiner reachability shortcut with $\widetilde{O}(m)$ additional edges and diameter $\text{polylog}(n)$, so \Cref{conj:Steiner} is true as well.

    What remains is to prove that \Cref{conj:Steiner} implies \Cref{conj:circuit}. If \Cref{conj:Steiner} is true, then by \Cref{lem:or_circuit} every Boolean function $f:\{0, 1\}^n \mapsto \{0, 1\}^n$ that can be computed by a Boolean OR-circuit $C$ with size $s_1$ can also be computed by an OR-circuit with size  $s_2 = \widetilde{O}(s_1)$  and diameter $d = O(\text{polylog}(n))$. We will now convert circuit $C$ into a circuit $C^*$ with \textit{depth} $O(d)$ by a simple layering operation. Let $S$ denote the set of input gates of $C$, let $T$ denote the output gates of $C$, and let $V$ denote the set of OR gates in $C$. We will make $d$ copies of set $V$, which we denote $V_1, \dots, V_d$. Sets $V_1, \dots, V_d$ will be the new OR gates in our new OR-circuit $C^*$. For each $v \in V$, we let $v_i \in V_i$ denote the copy of $v$ in $V_i$.    We define the edge set of circuit $C^*$ as follows:
    \begin{itemize}
        \item For each edge $(x, y) \in S \times V$ in circuit $C$, add edge $(x, y_1) \in S \times V_1$ to $C^*$.
        \item For each edge $(x, y) \in V \times T$ in circuit $C$, add edge $(x_d, y) \in V_d \times T$ to $C^*$.
        \item For each edge $(x, y) \in S \times T$ in circuit $C$, add edge $(x, y) \in S \times T$ to $C^*$.
        \item For each edge $(x, y) \in V \times V$, add  edge $(x_i, y_{i+1}) \in V_i \times V_{i+1}$ to $C^*$ for every $i \in [1, d-1]$. 
        \item For each $v \in V$ and each $i \in [1, d-1]$, add the edge $(v_i, v_{i+1}) \in V_i \times V_{i+1}$ to $C^*$. 
    \end{itemize}
This completes the definition of the new circuit $C^*$. 

We now quickly verify that circuit $C^*$ computes the same function as $C$. Notice that an input gate $s \in S$ can reach an output gate $t \in T$ in circuit $C$ if and only if there is an $(s, t)$-path $P = (s, v^1, v^2, \dots, v^k, t)$ in $C$ of length at most $d$. If $P$ exists in $C$, then circuit $C^*$ must contain the $(s, t)$-path $P^* = (s, v^1_1, v^2_2, \dots, v^k_k, v^k_{k+1}, \dots, v^k_d, t)$. Likewise, if there exists an $(s, t)$-path $P^*$ in $C^*$, then there exists an $(s, t)$-path $P$ in $C$.
Then by \Cref{obs:or_reach}, we conclude that circuit $C$ and circuit $C^*$ compute the same binary function. 

Finally, observe that circuit $C^*$ has depth $O(d) = O(\text{polylog}(n))$ since every path in $C^*$ can pass through at most one vertex in each layer $V_i$ for $i \in [1, d]$. Likewise, circuit $C^*$ has size at most $s_2 \le s_1d = \widetilde{O}(s_1)$, since every edge in circuit $C$ is copied at most $d$ times in $C^*$. 
\end{proof}

\section{Steiner Shortcut for the Lower Bound of \cite{WilliamsXX24}}
\label{sec:app:breaking}

\subsection{Construction of \cite{WilliamsXX24} Lower Bound} 

Before presenting our Steiner shortcut for the lower bound construction of \cite{WilliamsXX24}, we first give a sketch of its construction and relevant properties.

The lower bound of \cite{WilliamsXX24} consists of a graph $G$ and an associated collection of critical paths $\mathcal{P}$.  Graph $G$ is parameterized by an integer $r$. Integer $r$ will be chosen to be a power of two. 
\begin{itemize} 
    \item \textbf{Vertices $V(G)$.} Graph $G$ has $2r+1$ layers, $L_0, \dots, L_{2r}$. The vertices of each layer can be written as a collection of points on a three-dimensional grid. Formally, for $i \in [0, 2r]$,
    $$
    L_i = \{i\} \times [1, 4r] \times [1, 4r] \times [1, 4r^2].
    $$
    \item \textbf{Edges $E(G)$.} Before constructing the edge set of $G$, we define a permutation $\sigma$ of $[0, r-1]$ as follows. For any  $i \in [0, r-1]$, we define the $i$th element $\sigma_i$ of $\sigma$ as the integer whose corresponding $\lg(r)$-bit binary representation can be reversed to obtain the $\lg(r)$-bit binary representation of $i$. Let $v_0, \dots, v_{2r-1}$ be a sequence of vectors where
    $$
    v_i = (1, 0, \sigma_{i/2 })
    $$
    for every even $i \in [0, 2r-1]$, and $$v_i = (0, 1, \sigma_{(i-1)/2})$$ for every odd $i \in [0, 2r-1]$.

    For every $i \in [0, 2r-1]$ and every point $x \in [4r] \times [4r] \times [4r^2]$, we add the edge $(i, x) \rightarrow (i+1, x)$ to edge set $E(G)$, and we add the edge $(i, x) \rightarrow (i+1, x+v_i)$ to $E(G)$ if $(i+1, x+v_i) \in L_{i+1}$ is a valid point in layer $L_{i+1}$.\footnote{The edge set $E(G)$ is constructed using permutation $\sigma$ to ensure Property 3 of Lemma \ref{lem:properties_williams}.}

    \item \textbf{Critical Paths $\mathcal{P}$.} We define a collection $\mathcal{P}$ of paths in graph $G$. 
    Towards this goal, we first define a set of source nodes $S \subseteq L_0$. Specifically, we choose $S$ to be
    $$
    S = \{0\} \times [2r] \times [2r] \times [2r^2] \subseteq L_0.
    $$

    For each node  $s \in S$ and each choice of tuple $(d_1, d_2) \in [r] \times [r]$, we will define a critical path $P^s_{d_1, d_2} \in \mathcal{P}$ that starts at node $s \in L_0$ and ends at a node in $L_{2r}$. Before we can define these critical paths, we will need to define a function $f_{d}:[0, r-1] \mapsto [0, r]$ for every $d \in [r]$. We define $f_{d}$ as follows.
\[  f_{d}(i) =  \begin{cases} 
      \sigma_{i} & \sigma_{i} < d  \\
      0 & \text{else}
   \end{cases}
\]

We are now ready to define our critical paths $\mathcal{P}$. Fix a node $s \in S$ and a tuple $(d_1, d_2) \in [r] \times [r]$. We iteratively define a path $P^s_{d_1, d_2}$ starting at $s$ and ending at a node in $L_{2r}$. We now specify the edge that path $P^s_{d_1, d_2}$ takes between layers $L_i$ and $L_{i+1}$, where $i \in [0, 2r-1]$. 
\begin{itemize}
    \item If $i \equiv 0 \mod 2$, then path $P^s_{d_1, d_2}$ takes an edge of the form 
    $$
    (i, x, y, z) \rightarrow (i+1, x+1, y, z + f_{d_1}(i/2)).
    $$
    \item If $i \equiv 1 \mod 2$, then path $P^s_{d_1, d_2}$ takes an edge of the form
    $$
    (i, x, y, z) \rightarrow (i+1, x, y+1, z+f_{d_2}((i-1)/2).
    $$
\end{itemize}

    We have established that path $P^s_{d_1, d_2}$ starts at node $s \in L_0$, and we have established which type of edge $P^s_{d_1, d_2}$ takes between layers $L_i$ and $L_{i+1}$ where $i \in [0, 2r-1]$, so we have uniquely specified path $P$. We define $\mathcal{P}$ to be the collection of paths $\{P^s_{d_1, d_2}\}_{s, d_1, d_2}$  defined over all choices of $s \in S$ and $d_1, d_2 \in [1, r]$.%
\end{itemize}
We state without proof several properties of this graph construction.
\begin{lemma}[cf. Proposition 5.1 of \cite{WilliamsXX24}] Graph $G$ and paths $\mathcal{P}$ have the following properties:
\begin{enumerate}
    \item $|V(G)|=\Theta(r^5)$, $|E(G)| = \Theta(r^5)$, and $|\mathcal{P}|=\Theta(r^6)$.
    \item (Uniqueness.) Every critical path in $\mathcal{P}$ is the unique path between its endpoints in $G$.
    \item (Low Overlap.) Any path $\sigma$ in $G$ of length at most $s \leq 2r$ is contained in at most $O((r/s)^2)$ distinct critical paths in $\mathcal{P}$.
\end{enumerate}
\label{lem:properties_williams}
\end{lemma}

This graph construction of \cite{WilliamsXX24} implies a polynomial lower bound on the diameter of $O(m)$-size shortcuts. We summarize this property of $G, \mathcal{P}$ in the following theorem. 

\begin{theorem}[cf. Theorem 1.3 of \cite{WilliamsXX24}]
\label{thm:w_lb_thm} Fix a constant $c>0$. 
Let $H \subseteq TC(G)$ be a shortcut set of $G$ of size $|H| \leq c \cdot |E(G)|$. If construction parameter $r$ of $G, \mathcal{P}$ is chosen to be sufficiently large, then there must exist an $s \leadsto t$ path $P \in \mathcal{P}$ such that $\text{\normalfont dist}_{G \cup H}(s, t) = \Omega(r) =  \Omega(|V(G)|^{1/5})$. 
Consequently, for any $O(m)$-size shortcut set $H \subseteq TC(G)$, graph $G\cup H$ has diameter $\Omega(n^{1/5})$. 
\end{theorem}
The proof of Theorem \ref{thm:w_lb_thm} uses a potential function argument which shows that every shortcut edge in $H$ can only decrease the lengths of paths in $\mathcal{P}$ by a small amount; intuitively, this follows from the low overlap property of paths in $\mathcal{P}$ stated in Lemma \ref{lem:properties_williams}.

\subsection{Steiner Shortcut to \cite{WilliamsXX24}}
We will now give a Steiner shortcut for the lower bound construction of \cite{WilliamsXX24}. This Steiner shortcut rules out the possibility of extending Theorem \ref{thm:w_lb_thm} to Steiner reachability shortcut sets via the argument used in \cite{WilliamsXX24}. Formally, we will prove the following theorem about the graph $G$ and set of paths $\mathcal{P}$ corresponding to the lower bound construction of \cite{WilliamsXX24}.

\begin{theorem}[Steiner Shortcut for \cite{WilliamsXX24} Construction]
Graph $G$ and set of paths $\mathcal{P}$ admit a Steiner reachability shortcut set $H$ such that
\begin{enumerate}
    \item $H$ has $O(r^4 ) = O(|V(G)|/r)$ Steiner nodes and $O(r^5) = O(|E(G)|)$ edges, and 
    \item For every $s \leadsto t$ path $P \in \mathcal{P}$, there exists an $s \leadsto t$ path of length two in $G \cup H$.
\end{enumerate}
\label{thm:attack2}
\end{theorem}

\paragraph{Construction of $H$.}
The additional Steiner nodes in our Steiner reachability shortcut set $H$ will correspond to a collection of points in a three-dimensional grid. Specifically, we define a set of Steiner nodes $W$ to be
$$
W = \{-1\} \times \{1, 4r\} \times \{1, 4r\} \times [1, 4r^2].
$$
(We add the $-1$ in the first coordinate of $W$ to distinguish the vertices in $W$ from the vertices in $V(G)$.) The edge set $E(H)$ of our Steiner shortcut will satisfy  $E(H) \subseteq (L_0 \times W) \cup (W \times L_{2r})$.

Before we can define $E(H)$, we need to introduce some new definitions. For every $d \in [r]$ and every $i \in \{1, 2\}$, we define integer $z^i_{d} \geq 0$ as follows:
$$
z^1_{d} = \sum_{j \in [0, r-1]} f_{d}(2j) \text{ \quad and \quad }  z^2_{d} = \sum_{j \in [0, r-1]} f_{d}(2j+1).
$$

Using integers $z^i_{d}$, we are now ready to construct $E(H)$.
\begin{itemize}
    \item For every node $(0, x, y, z) \in S$ and every $d \in [r]$, add the edge 
    $$
    (0, x, y, z) \rightarrow \left(-1, x+r, y, z+z^1_d\right)
    $$
    to $E(H)$ if $\left(-1, x+r, y, z+z^1_d\right) \in W$ is a valid point in $W$. 
    \item For every node $(-1, x, y, z) \in W$ and every $d \in [r]$, add the edge
    $$
    (-1, x, y, z) \rightarrow \left(2r, x, y+r, z+z^2_d\right)
    $$
    to $E(H)$ if $\left(2r, x, y+r, z+z^2_d\right) \in L_{2r}$ is a valid node in layer $L_{2r}$. 
\end{itemize}
This completes the construction of Steiner shortcut set $H$ and additional edges $E(H)$. 

\paragraph{Analysis of $H$.}
First we will bound the number of Steiner nodes in $W$ and the number of additional edges $E(H)$ in our Steiner shortcut set. We observe that $|W| = 4^3r^4 = O(r^4)$, as claimed. Additionally, we observe that the total number of edges in our Steiner shortcut is at most 
$$
|E(H)| \leq |S| \cdot r + |W| \cdot d \leq 2^7 r^5 = O(r^5),
$$
as claimed.

We next verify that $H$ is a valid Steiner reachability shortcut set, i.e., that the reachability relation between nodes in $V(G)$ remains the same in $G \cup H$. Let $u, v \in V(G)$ be nodes such that path $(u, w, v)$ exists in graph $G \cup H$ for some $w \in W$. Then we claim that $u$ can reach $v$ in $G$ as well. Let $u = (0, x, y, z)$. We can write $u \rightarrow w \rightarrow v$ as
$$
(0, x, y, z) \rightarrow (-1, x+r, y, z+z^1_{d_1}) \rightarrow (2r, x+r, y+r, z+z^1_{d_1}+z^2_{d_2}), 
$$
for some choice of $d_1, d_2 \in [1, r]$. We claim that critical path $P^u_{d_1, d_2} \in \mathcal{P}$ is a $u \leadsto v$ path in $G$. Notice that by construction, path $P^u_{d_1, d_2}$ begins at node $u$ and ends at node $v'$, where
$$
v' = \left(2r, x+r, y+r, z + \sum_{i=0}^{r-1}f_{d_1}(i) + \sum_{i=0}^{r-1}f_{d_2}(i)   \right) = (2r, x+r, y+r, z + z_{d_1}^1 +z_{d_2}^2 ) = v.
$$
We conclude that  $P^u_{d_1, d_2}$ is a $u \leadsto v$ path in $G$, as claimed. 

Finally, we verify that for every $s \leadsto t$ path $P^s_{d_1, d_2} \in \mathcal{P}$, there exists an $s \leadsto t$ path in $G \cup H$ of length two. Notice that by our construction of $\mathcal{P}$, if  $P^s_{d_1, d_2}$ is an $s \leadsto t$ path, then nodes $s$ and $t$ can be written as
$$
s = (0, x, y, z) \text{\quad and \quad} t = \left(  2r, x+r, y+r, z+z_{d_1}^1 + z_{d_2}^2 \right). 
$$
Then we can identify the following $s \leadsto t$ path in $G \cup H$ of length two:
$$
(0, x, y, z) \rightarrow (-1, x+r, y, z+z_{d_1}^1) \rightarrow \left(  2r, x+r, y+r, z+z_{d_1}^1 + z_{d_2}^2 \right).
$$
This completes the proof of Theorem \ref{thm:attack2}.

\section{Omitted Proofs}

\subsection{Proof of \Cref{lem:criticalpairbasic}}
\label{app:criticalpairbasic}
\criticalpairbasic*
\begin{proof}
      Each vertex $((0,0),(p_0,\ldots,p_{k-1})) \in S$  is part of 
    $\Delta^k$ critical pairs because for any critical direction 
    $(\vec{v}_0, \ldots, \vec{v}_{k-1}) \in \mathcal{B}(r, d)^k$, 
    we have that   
    $((\ell,0),(p_0+\ell \vec{v}_0,\ldots,p_{k-1}+\ell\vec{v}_{k-1})) \in L_{\ell, 0}$ is a valid vertex in $L_{\ell,0}$, since $p_j+\ell\vec{v}_j \in (-(r+1)r\ell, (r+1)r\ell]^d$.
    Additionally, the number of vertices in $S$ is at least $|S|$ = $(2r^2\ell)^{dk} = \frac{N r^{dk}}{(r+1)^{dk}}$. 
    We conclude that the total number of critical 
    pairs is $\frac{N \Delta^k r^{dk}}{(r+1)^{dk}}$.
    
    Turning to the uniqueness claim, if $(s, t) \in P$ is 
    an arbitrary critical pair we can write
    \begin{align*}
    s &= ((0, 0), (p_0, \ldots, p_{k-1})) \in S\\
    t &= ((\ell,0),(p_0+\ell \vec{v}_0,\ldots,p_{k-1}+\ell\vec{v}_{k-1})) \in L_{\ell,0},
    \end{align*}
    One $s$-$t$ path $P_{s,t}$ passes through all vertices of the form
    \[
    ((i, j), (p_0 + (i+1)\vec{v}_0, \ldots, p_{j-1} + (i+1)\vec{v}_{j-1}, p_j + i\vec{v}_j, \dots, p_{k-1} + i\vec{v}_{k-1})).
    \]
    We claim that $P_{s, t}$ is the unique 
    $s$-$t$ path in $G$. 
    By construction, any $s$-$t$ path $P'$ in $G$ can be identified with a 
    sequence of $k\ell$ vectors $\vec{u}_{0,0},\ldots,\vec{u}_{\ell-1,k-1}$ corresponding to the edges of $P'$ that satisfies
    \[
    \forall j\in [0,k),\;\;  \sum_{i=0}^{\ell-1} \vec{u}_{i,j} = \ell\vec{v}_{j}.
    \]
    However, if any $\vec{u}_{i,j}\neq \vec{v}_j$ then $(\sum_i \vec{u}_{i,j})/\ell$ is not an extreme point on the convex hull 
    $\cB(r,d)$, and therefore cannot be equal to $\vec{v}_j$.
    Thus $\vec{u}_{i,j}=\vec{v}_j$ and $P'=P_{s,t}$.

    Since each critical pair is connected by a unique path, we can refer to
    it as a \emph{critical path}.
    The last claim follows since any $k$ consecutive edges along a critical path uniquely determines its critical direction $(\vec{v}_0,\ldots,\vec{v}_{k-1})$, which uniquely determines its first vertex in $L_{0,0}$ and last vertex in $L_{\ell,0}$.
\end{proof}

\end{document}